\documentclass[%
aps,
prl,
reprint,
superscriptaddress,
nofootinbib,
amsmath,
amssymb
]{revtex4-2}
\usepackage{
    physics,
    graphicx,
    multirow,
    tabularx,
    mathtools, 
    placeins, 
    nicefrac, 
    hyperref, 
    threeparttable, 
    siunitx,
    enumitem,
    makecell,
    amsmath,amssymb,amsthm,mathrsfs,amsfonts,dsfont
}
\usepackage{soul}
\usepackage[capitalize]{cleveref}
\Crefname{section}{Sec.}{Secs.}

\usepackage[dvipsnames]{xcolor}
\hypersetup{
    colorlinks,
    linkcolor={blue!50!black},
    citecolor={blue!50!black},
    urlcolor={blue!80!black}
}

\usepackage[caption=false]{subfig} 
\newtheorem{observation}{Observation}

\newtheorem{definition}{Definition}

\begin{document}

\title{Entanglement Superactivation in Multiphoton Distillation Networks}

\author{Rui Zhang}
\altaffiliation{These authors contributed equally to this work.}
\affiliation{Hefei National Research Center for Physical Sciences at the Microscale and School of Physical Sciences, University of Science and Technology of China, Hefei 230026, China}
\affiliation{Shanghai Research Center for Quantum Science and CAS Center for Excellence in Quantum Information and Quantum Physics, University of Science and Technology of China, Shanghai 201315, China}

\author{Yue-Yang Fei}
\altaffiliation{These authors contributed equally to this work.}
\affiliation{Hefei National Research Center for Physical Sciences at the Microscale and School of Physical Sciences, University of Science and Technology of China, Hefei 230026, China}
\affiliation{Shanghai Research Center for Quantum Science and CAS Center for Excellence in Quantum Information and Quantum Physics, University of Science and Technology of China, Shanghai 201315, China}

\author{Zhenhuan Liu}
\altaffiliation{These authors contributed equally to this work.}
\affiliation{Center for Quantum Information, Institute for Interdisciplinary Information Sciences, Tsinghua University, Beijing 100084, China}

\author{Xingjian Zhang}
\affiliation{Center for Quantum Information, Institute for Interdisciplinary Information Sciences, Tsinghua University, Beijing 100084, China}

\author{Xu-Fei Yin}
\affiliation{Hefei National Research Center for Physical Sciences at the Microscale and School of Physical Sciences, University of Science and Technology of China, Hefei 230026, China}
\affiliation{Shanghai Research Center for Quantum Science and CAS Center for Excellence in Quantum Information and Quantum Physics, University of Science and Technology of China, Shanghai 201315, China}
\affiliation{Hefei National Laboratory, University of Science and Technology of China, Hefei 230088, China}

\author{Yingqiu Mao}
\affiliation{Hefei National Research Center for Physical Sciences at the Microscale and School of Physical Sciences, University of Science and Technology of China, Hefei 230026, China}
\affiliation{Shanghai Research Center for Quantum Science and CAS Center for Excellence in Quantum Information and Quantum Physics, University of Science and Technology of China, Shanghai 201315, China}
\affiliation{Hefei National Laboratory, University of Science and Technology of China, Hefei 230088, China}

\author{Li Li}
\affiliation{Hefei National Research Center for Physical Sciences at the Microscale and School of Physical Sciences, University of Science and Technology of China, Hefei 230026, China}
\affiliation{Shanghai Research Center for Quantum Science and CAS Center for Excellence in Quantum Information and Quantum Physics, University of Science and Technology of China, Shanghai 201315, China}
\affiliation{Hefei National Laboratory, University of Science and Technology of China, Hefei 230088, China}

\author{Nai-Le Liu}
\affiliation{Hefei National Research Center for Physical Sciences at the Microscale and School of Physical Sciences, University of Science and Technology of China, Hefei 230026, China}
\affiliation{Shanghai Research Center for Quantum Science and CAS Center for Excellence in Quantum Information and Quantum Physics, University of Science and Technology of China, Shanghai 201315, China}
\affiliation{Hefei National Laboratory, University of Science and Technology of China, Hefei 230088, China}

\author{Otfried Gühne}
\affiliation{Naturwissenschaftlich-Technische Fakult\"at, Universit\"at Siegen, Walter-Flex-Stra{\ss}e 3, 57068 Siegen, Germany}

\author{Xiongfeng Ma}
\affiliation{Center for Quantum Information, Institute for Interdisciplinary Information Sciences, Tsinghua University, Beijing 100084, China}
\affiliation{Hefei National Laboratory, University of Science and Technology of China, Hefei 230088, China}

\author{Yu-Ao Chen}
\affiliation{Hefei National Research Center for Physical Sciences at the Microscale and School of Physical Sciences, University of Science and Technology of China, Hefei 230026, China}
\affiliation{Shanghai Research Center for Quantum Science and CAS Center for Excellence in Quantum Information and Quantum Physics, University of Science and Technology of China, Shanghai 201315, China}
\affiliation{Hefei National Laboratory, University of Science and Technology of China, Hefei 230088, China}

\author{Jian-Wei Pan}
\affiliation{Hefei National Research Center for Physical Sciences at the Microscale and School of Physical Sciences, University of Science and Technology of China, Hefei 230026, China}
\affiliation{Shanghai Research Center for Quantum Science and CAS Center for Excellence in Quantum Information and Quantum Physics, University of Science and Technology of China, Shanghai 201315, China}
\affiliation{Hefei National Laboratory, University of Science and Technology of China, Hefei 230088, China}

\begin{abstract}
	In quantum networks, after passing through noisy channels or information processing, residual states may lack sufficient entanglement for further tasks, yet 
    they may retain hidden quantum resources that 
    can be recycled. Efficiently recycling these states to extract entanglement resources such as genuine multipartite entanglement or Einstein-Podolsky-Rosen pairs is essential for optimizing network performance. 
	Here, we develop a tripartite entanglement distillation scheme using an eight-photon quantum platform, demonstrating entanglement superactivation phenomena which are unique to multipartite systems. 
	We successfully generate a three-photon genuinely entangled state from two bi-separable states via local operations and classical communication, demonstrating superactivation of genuine multipartite entanglement.
	Furthermore, we extend our scheme to generate a three-photon state capable of extracting an Einstein-Podolsky-Rosen pair from two initial states lacking this capability, revealing a previously unobserved entanglement superactivation phenomenon. 
	Our methods and findings offer not only practical applications 
    for quantum networks, but also lead to a deeper understanding of multipartite entanglement structures.
\end{abstract}

\maketitle

\noindent \emph{Introduction.---}
In the future, quantum devices are expected to form a global network, connecting remote clients through entanglement~\cite{Simon2017network,Storz2023loophole,sahu2023micro}, likely facilitated by photon transmission~\cite{pan2012multiphoton}. Clients harness entanglement to carry out various information processing tasks, such as quantum key distribution~\cite{bennett1984quantum,ekert1991quantum}, quantum teleportation~\cite{bennett1993teleporting}, and blind quantum computation~\cite{barz2012blind}. 
Upon transmission through noisy quantum channels or following the completion of information processing procedures, however, the residual states may lack the necessary quantum entanglement for subsequent quantum tasks.
Effectively recycling and converting these residual states into valuable quantum entanglement resources is key to enhancing network efficiency.

Entanglement distillation, which transforms weakly entangled states 
into highly entangled ones, is a typical recycling technique for quantum 
resources. While bipartite distillation protocols have been studied 
from theoretical~\cite{bennett1996purification, DeutschPRL1996, Pan2001purification} and experimental
~\cite{Pan2003purification,chen2017experimental, reichle2006experimental,hu2021long, ecker2021experimental, 
SalartPRL2010, Hanson2017entanglement, YanPRL2022} 
points of view, the study of multipartite protocols has, to the 
best of our knowledge, been carried out in theory only \cite{MuraoPRA1998, huber2011purification, DurPRL2003Multiparticle,MiyakePRL2005, DurPRA2006, OtfriedPRA2023, ahmad2025distillation}. Interestingly, if one aims
at distilling genuine multipartite entanglement (GME), the 
phenomenon of superactivation can help. This phenomenon describes the 
fact that starting from a given quantum state without GME, multiple copies of it may 
retain GME~\cite{huber2011purification,Yamasaki2022activation}. For bipartite
systems, similar phenomena only exist for Bell nonlocality~\cite{SuperactivationPRL2012} and quantum steering~\cite{SuperactivationPRA2016}. So, for multiparticle systems, 
entanglement superactivation promises one to harvest something useful out 
of ``nothing'', thereby enhancing the overall efficiency of quantum 
networks.

An effective multipartite entanglement distillation scheme demands sophisticated multiphoton entanglement preparation and manipulation techniques.
Currently, the spontaneous parametric down-conversion (SPDC) process stands out as the most mature technique for preparing multiphoton entangled states~\cite{pan2012multiphoton}.
However, its inherent probabilistic nature results in double-pair emissions, which introduce spurious contributions 
to experimental results and pose a notorious obstacle to faithfully  demonstrate entanglement distillation protocols~\cite{Pan2003purification,chen2017experimental}.

In this work, by delicately designing the entanglement distillation network into a crossed structure, we realize a tripartite entanglement distillation scheme which can filter out unwanted double-pair emission noises. 
Applying this scheme, we experimentally realize multipartite entanglement distillation for the first time and demonstrate two types of entanglement superactivation phenomena as sketched in Fig.~\ref{fig:overview}. 
First, we experimentally generate a three-photon genuinely tripartite entangled state from two copies of three-photon states without GME, demonstrating GME superactivation. 
Second, we define the \emph{stochastic localizable entanglement} (SLE) to describe the ability to extract local entanglement from given multipartite states~\cite{Chitambar2014LOCC,popp2005localizable} and theoretically predict the SLE superactivation phenomenon.
Furthermore, we show that there exist states whose GME can be activated while their SLE remains non-activable, highlighting the fundamental physical implications of SLE superactivation. 
By applying an additional single-photon measurement to the tripartite distillation scheme, we experimentally observe the SLE superactivation phenomenon.

\begin{figure}[t]
	\centering
	\includegraphics[width=1 \linewidth]{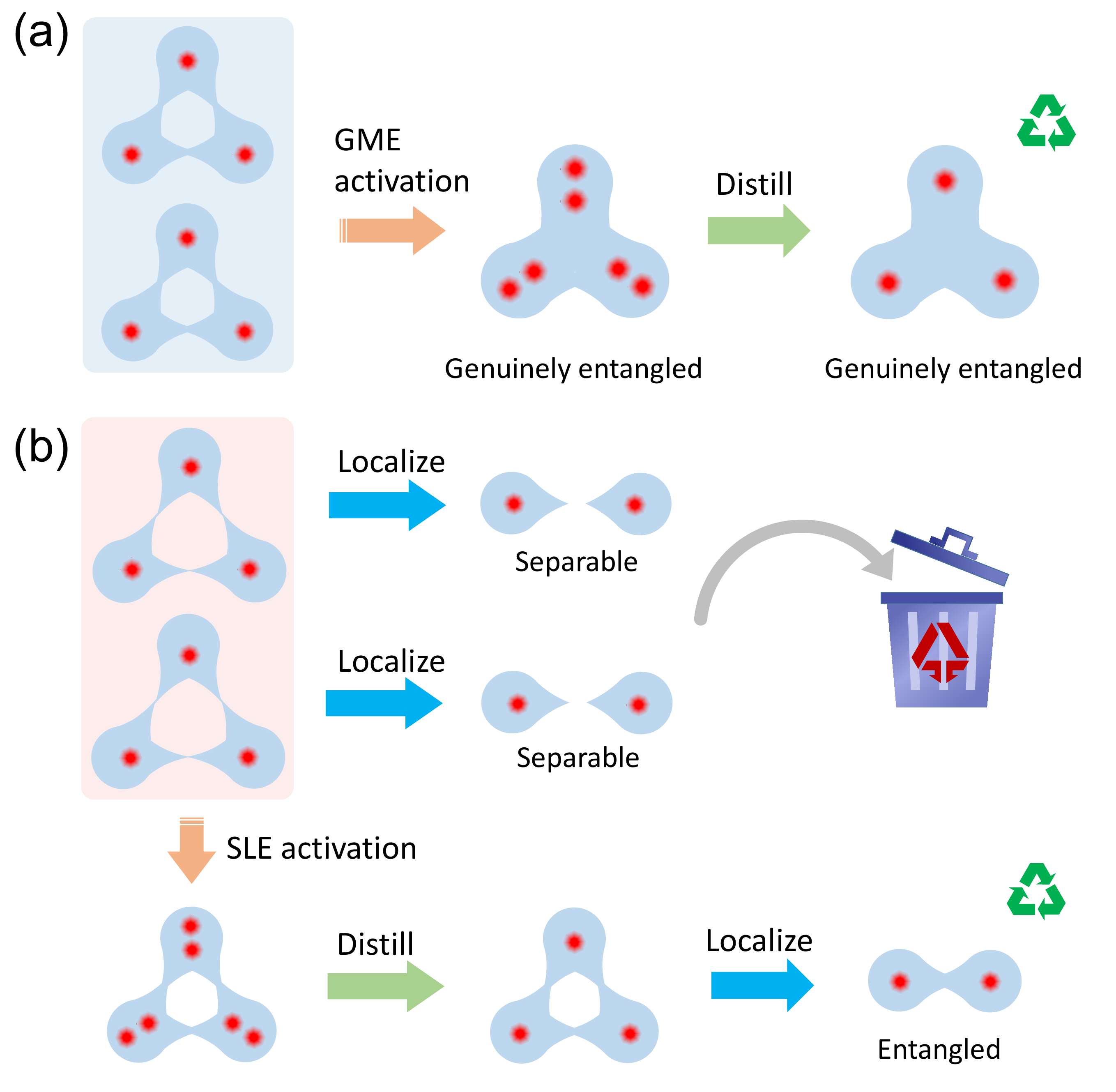}
	\caption{Schematic description of two types of entanglement superactivation in tripartite systems. Red dots represent photons.
	(a) GME superactivation. 
	GME can be activated by collecting two biseparable tripartite states to form a new tripartite state (orange arrow). 
	One can then perform tripartite distillation (green arrow) to concentrate the resources.   
	(b) SLE superactivation. SLE is the property of some quantum 
    states that bipartite entanglement can be localized between
    two parties with any SLOCC protocol.
	For some tripartite states, any localization protocol results in bipartite separable states shared between 
    two subsystems (blue arrow) and the resulting states 
    cannot be used to distill entanglement.
	When two copies of such tripartite states are collected, SLE can be activated and it becomes possible to create a bipartite entangled state. 
	}
	\label{fig:overview}
\end{figure}

\begin{figure*}[htbp!]
	\centering
	\includegraphics[width=1 \linewidth]{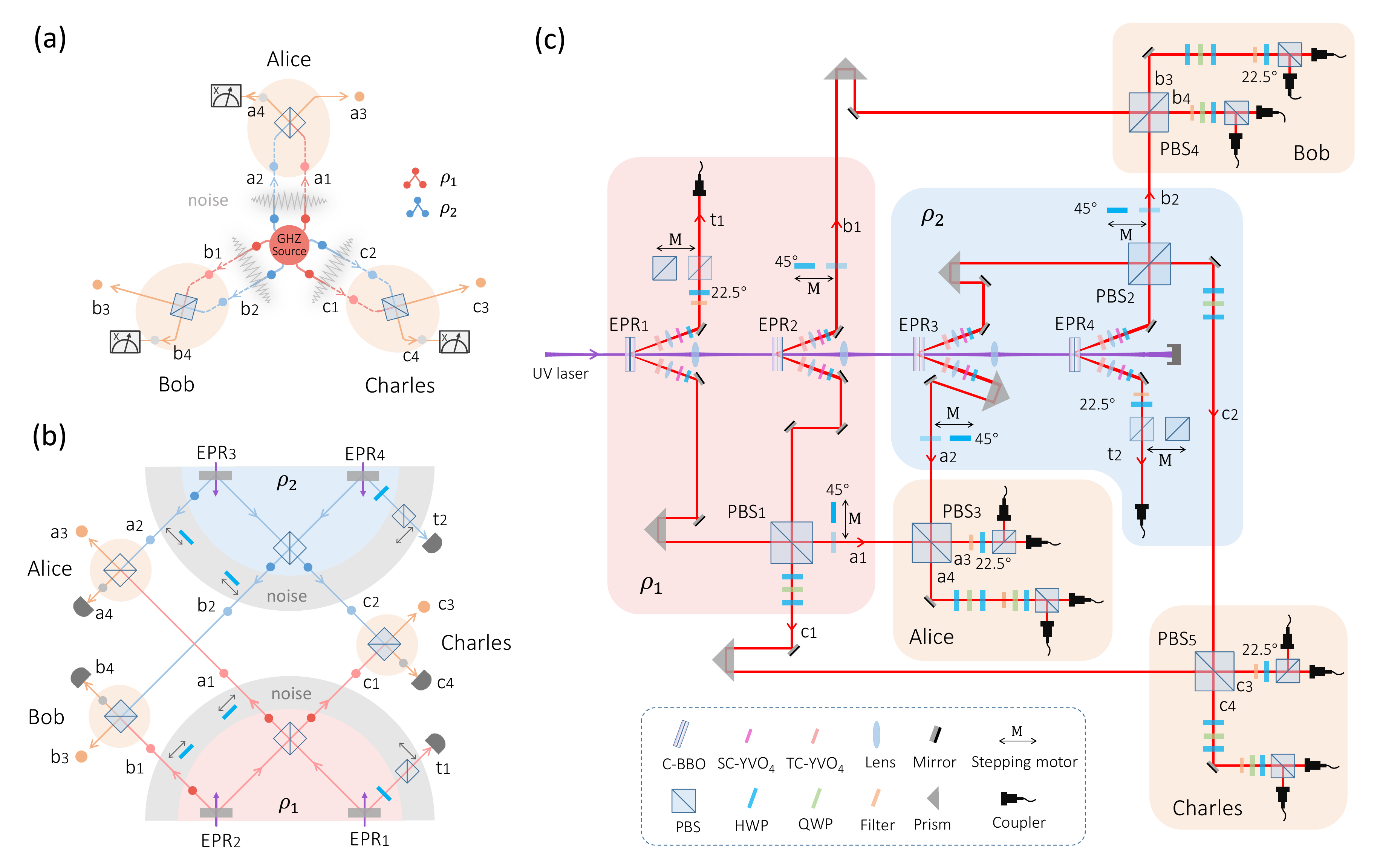}
	\caption{Schematic description of the tripartite distillation scheme and the experimental setup. 
		(a) Schematic of the tripartite entanglement distillation scheme.
		(b) Schematic of the noisy GHZ state generation procedure and entanglement distillation network.
		(c) Detailed experimental setup. 
		Each to-be-distilled noisy GHZ state is independently prepared with two EPR sources.
		Three stepping motors are employed to drive a PBS and two HWPs to prepare different components of the noisy GHZ state.
		PBS$_3$, PBS$_4$, and PBS$_5$ are employed to realize the distillation operation by overlapping photons from two to-be-distilled states.
		The QWP-HWP-QWP combination after PBS is used to compensate for phase drift between two photons after interference on PBS. 
		C-BBO: combination of $\beta$-barium borate crystals; 
		SC-YVO$_4$: YVO$_4$ crystal for spatial compensation;
		TC-YVO$_4$: YVO$_4$ crystal for temporal compensation;
		HWP: half-wave plate; 
		QWP: quarter-wave plate.
	}
	\label{Fig:setup}
\end{figure*}

\noindent\emph{Entanglement superactivation.---}
GME represents one of the most valuable resources in quantum networks.
In short, a tripartite state is biseparable, if it can be decomposed as $\rho=
p_1 \rho^{\rm sep}_{A|BC}
+ p_2 \rho^{\rm sep}_{B|AC}
+ p_3 \rho^{\rm sep}_{C|AB}$, 
where the $p_i$ form a probability distribution and the 
$\rho^{\rm sep}_{X|YZ}$ are separable for the respective bipartition. 
States which are not biseparable have GME \cite{gunhe2009entanglement}. This property cannot be created via stochastic local operations assisted by classical communication (SLOCC) and possesses the property of GME superactivation, as illustrated in Fig.~\ref{fig:overview}(a).
In addition to GME, a multipartite state can possess other types of resources, like the ability to extract Einstein-Podolsky-Rosen (EPR) pairs, which is relevant for connecting different nodes in quantum networks.
Such ability can be formalized as a basic property for an arbitrary multipartite state defined as follows. 
\begin{definition}[Stochastic Localizable Entanglement]\label{def:SLE}
	A multipartite state $\rho$ possesses stochastic localizable entanglement (SLE) on subsystems $A$ and $B$ if and only if there exists a protocol with local operations, classical communication and post-selection on measurement results, (that is, a SLOCC operation) that transforms $\rho$ into a bipartite entangled state shared between $A$ and $B$.
\end{definition}
SLE extends concepts of entanglement of assistance~\cite{divincenzo1999assistance} and localizable entanglement~\cite{verstraete2004le,popp2005localizable}. 
A direct corollary of this definition is that, similar to GME, SLE cannot be created from any state without SLE through SLOCC, as its definition already considers all possible SLOCC protocols.
In the Section II.A of Supplementary Material~\cite{Supplemental}, we derive a criterion to reduce the search over all SLOCC protocols to those involving projective measurements only.
Based on this criterion, we theoretically prove that:
	\begin{observation}[SLE Superactivation]
		SLE can be activated by collecting multiple copies of a state without SLE. 
	\end{observation}
	\noindent As illustrated in Fig.~\ref{fig:overview}(b), there exists some state from which any SLOCC operation results in a separable bipartite state. 
	However, by collecting two copies of this state, it becomes possible to obtain a bipartite entangled state through entanglement distillation and localization, both of which are just some SLOCC protocols.

Although GME and SLE share many similar properties, they represent two distinct notions of multipartite entanglement without any inclusion relation. 
In Section IV of Supplementary Material~\cite{Supplemental}, we prove this by demonstrating the existence of states that exhibit only GME (but no SLE) or only SLE (and no GME)~\cite{gunhe2016emergent}. 
Furthermore, assuming the validity of the PPT square conjecture~\cite{ruskai2012operator,chen2019pptsquare}, we show that the superactivation of SLE on all bipartite subsystems, like the subsystems $A$ and $B$ in Definition~\ref{def:SLE}, is strictly harder than that of GME. 
While it has been proven that any multipartite state that is not separable across any fixed bipartition can have its GME activated by collecting sufficiently many copies~\cite{Palazuelos2022genuinemultipartite}, usage of the PPT square conjecture allows to identify states that are not separable in any fixed bipartition but whose SLE cannot be activated with any finite number of copies. 
Taken together, these results highlight the profound physical implications and the distinctive nature of SLE and its superactivation, leading to an intriguing observation: \emph{local entanglement is harder to activate than the global entanglement}.

Similar to quantum statistics in entanglement concentration~\cite{QuantumStatisticsPRL2002}, both GME and SLE superactivation can be realized simply by collecting multiple copies of the target state, thereby directly generating the activated entanglement in the multi-qubit system without requiring any further operations~\cite{huber2011purification, Yamasaki2022activation,weinbrenner2024superactivation}.
In practical applications, to harness the activated entanglement, it is essential to concentrate it into a state with fewer qubits, as shown in Fig.~\ref{fig:overview}. 
This is, in fact, not always possible~\cite{weinbrenner2024superactivation}, so for practical purposes, an efficient entanglement distillation protocol, as a crucial technique for extracting entanglement from multi-qubit states, is required.
On the other hand, note that both SLE and GME cannot be created from SLOCC operations and that entanglement distillation is comprised solely of SLOCC operations.
Therefore, entanglement distillation is also a way to experimentally verify the existence of entanglement superactivation phenomena.

\noindent \emph{Tripartite entanglement distillation scheme.---}
In practical quantum networks, the Greenberger-Horne-Zeilinger (GHZ) state, denoted as $\ket{\mathrm{GHZ}_3}=(\ket{000}+\ket{111})/{\sqrt{2}}$, is a vital resource for quantum network tasks.
	However, when transmitted through noisy channels, the distributed GHZ state may degrade into a mixed state, which would lack sufficient entanglement to perform the intended tasks.
Here, we consider the following typical noisy GHZ state,
\begin{equation}\label{eq:Werner}
	\rho_\mathrm{noise}=p\ketbra{\mathrm{GHZ}_3}{\mathrm{GHZ}_3}+(1-p)\frac{\mathbb{I}_3}{8},
\end{equation}
where $\frac{\mathbb{I}_3}{8}$ is the three-qubit maximally mixed state.
This noisy GHZ state contains various real-world potential errors, including bit-flip and phase-flip errors, and the parameter $p$ clearly quantifies the noise level of quantum channels. 
Theoretically, the parameter $p$ also characterizes entanglement properties of noisy GHZ states, which would benefit the experimental demonstration of entanglement superactivation~\cite{ollivier2001discord,hashemi2012genuine,ma2011gme}. 

In the experiment, by encoding $\ket{0}$ and $\ket{1}$ in the horizontal and vertical polarizations of photons, $\ket{H}$ and $\ket{V}$, the entanglement distillation can be achieved using the parity check function of polarizing beam splitter (PBS, transmitting $\ket{H}$ and reflecting $\ket{V}$)~\cite{Pan2001purification,Pan2003purification}.
As shown in Fig.~\ref{Fig:setup}(a), after distributing two copies of 
the initial state, labeled by a$_1$-b$_1$-c$_1$ and a$_2$-b$_2$-c$_2$, to three clients---Alice, Bob, and Charles---each introduces their photons to a PBS, retaining only events that both two photons emit from both outputs of the PBS. 
Each client then performs the Pauli-$X$ measurement on one of the two output photons, such as a$_4$-b$_4$-c$_4$, and compares measurement outcomes with others.
If an even number of states $\ket{-}=(\ket{H}-\ket{V})/\sqrt{2}$ is registered, they keep the remaining photons, a$_3$-b$_3$-c$_3$; otherwise, one of them applies the phase-flip operation to the remaining photon. 
As a result, the distilled state can acquire a higher entanglement than the initial states.

Figure~\ref{Fig:setup}(b) outlines the noisy GHZ state generation procedure and entanglement distillation network in the linear optical platform.
Each initial noisy GHZ state is prepared with two EPR photon pairs generated through the SPDC process. 
Taking the noisy GHZ state $\rho_{1}$ as an example, two photons from EPR$_1$ and EPR$_2$ are overlapped on PBS$_1$ to prepare a four-qubit GHZ state $\ket{\mathrm{GHZ}_4}=(\ket{H_{a_1}H_{b_1}H_{c_1}H_{t_1}}+\ket{V_{a_1}V_{b_1}V_{c_1}V_{t_1}})/\sqrt{2}$. 
Based on this state, a set of program-controlled stepping motors is employed to move optical elements—a PBS and two half-wave plates (HWPs)—in or out of light paths. 
This introduces flip operations to the photonic state, simulating the impact of white noise on the GHZ state. 
Consequently, the component states of the noisy GHZ state can be probabilistically prepared on photons a$_1$-b$_1$-c$_1$, allowing for the adjustment of the target parameter $p$.
In our experiment, the distillation network is carefully designed in a crossed structure. 
This effectively eliminates all same-order double-pair emission events of SPDC process, allowing for a faithful realization of tripartite entanglement distillation for two noisy GHZ states (see Section V.C of Supplementary Material~\cite{Supplemental} for details).

The detailed experimental setup is depicted in Fig.~\ref{Fig:setup}(c).
Four EPR sources are realized by sequentially targeting a pulsed ultraviolet laser (390 $\mathrm{nm}$, 80 $\mathrm{MHz}$, 150 $\mathrm{fs}$, 500 $\mathrm{mW}$) through four sandwich-like combinations of $\beta$-barium borate crystals~\cite{wang_experimental_2016,Li2019repeater}. 
By applying band-pass filters with $\Delta\lambda=4~\mathrm{nm}$ on each down-converted photons, the average counting rate for each EPR pair is registered as 199,000 $\mathrm{s}^{-1}$, with an average fidelity of 0.965.
Photons from different EPR sources are finely adjusted to achieve temporal and spatial overlaps on distillation PBSs, leading to Hong-Ou-Mandel-type interference. 
As a result, two $\ket{\mathrm{GHZ}_3}$ states prepared on photons a$_1$-b$_1$-c$_1$ and a$_2$-b$_2$-c$_2$ achieve fidelities of 0.833 and 0.837. 
The interference visibilities on PBS$_3$, PBS$_4$, and PBS$_5$ for the distillation operation are $79.6\%$, $78.8\%$, and $78.1\%$, respectively. 

\begin{figure*}[ht!]
	\centering
	\includegraphics[width=1 \linewidth]{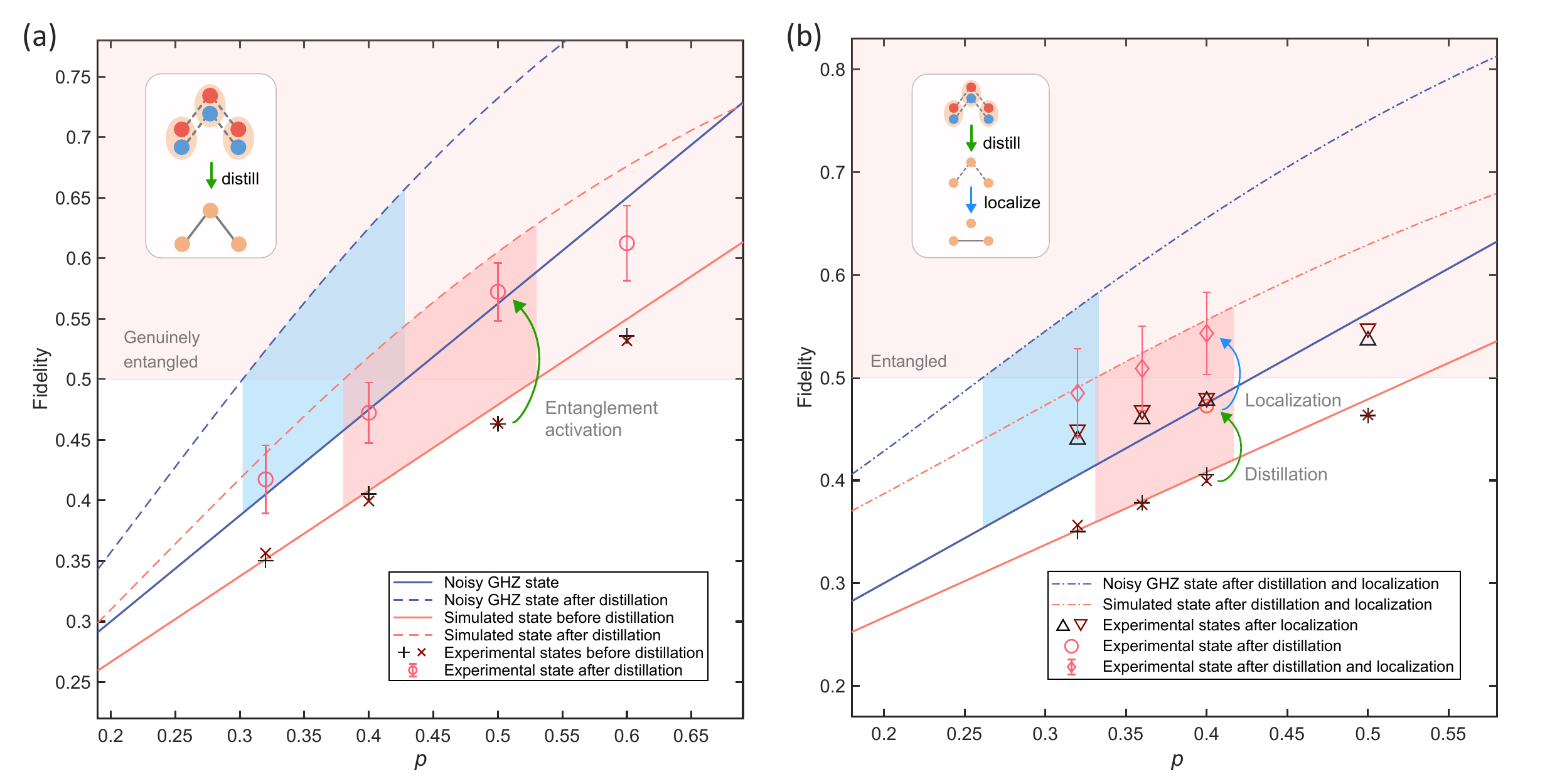}
	\caption{
		Experimental results. 
		Blue solid lines are fidelities between ideal noisy GHZ states and $\ket{\mathrm{GHZ}_3}$ state. 
		Red solid lines are fidelities of initial states simulated using the knowledge of noises in state preparation. 
		Blue and red shadowed regions represent ranges of entanglement superactivation.
		(a) Experimental results of the tripartite GME superactivation. 
		Blue and red dashed lines are calculated in the assumption of perfect distillation. 
		Due to unpredictable noises in the state preparation and distillation, experimental results have certain deviations from simulated lines. 
		GME superactivation is demonstrated by states with $p=0.5$, where the fidelity exceeds the GME threshold after distillation.
		(b) Experimental results of the tripartite SLE superactivation.
		Blue and red dash-dotted lines are calculated in the assumption of perfect distillation and localization with Pauli-$X$ measurement. 
		Triangles are fidelities of two-photon states extracted by individually localizing two initial states.
		Diamonds are fidelities of two-photon states extracted by tripartite distillation followed by localization. 
		For two-photon states after the localization operation, their $y$-coordinate values represent fidelities with the EPR pair.
		States with $p=0.5$ show that the localization operation can extract EPR pairs from noisy states without GME. 
		States with $p=0.4$ and $p=0.36$ show the existence of states that can be used to extract EPR pairs only with the assistance of tripartite distillation. 
	}
	\label{Fig:data}
\end{figure*}

\noindent \emph{Experimental results.---}
With the above experimental setup, we first evaluate the performance of our entanglement distillation scheme.
For noisy GHZ states, their GME can be determined using the entanglement witness~\cite{gunhe2009entanglement}, 
\begin{equation}\label{eq:EW}
	W = \frac{1}{2}\mathbb{I}_3-\ketbra{\mathrm{GHZ}_3},
\end{equation} 
which is equivalent to checking whether the fidelity $F=\bra{\mathrm{GHZ}_3}\rho\ket{\mathrm{GHZ}_3}$ exceeds $0.5$. 
Thus, we adopt $F$ as the target quantity to evaluate the performance of our entanglement distillation scheme.
Figure~\ref{Fig:data}(a) presents GME distillation results of noisy GHZ states with varying values of $p$. 
It exhibits a significant improvement in fidelity, demonstrating the effectiveness of our entanglement distillation scheme.
For instance, for states with a target parameter of $p=0.6$, we successfully extract a highly entangled state with a fidelity of 0.612(31) from two weakly entangled initial states with
actual fidelities of 0.536 and 0.532.


For states with a target parameter of $p=0.5$, the fidelities of the two initial states are $0.463$ and $0.464$, suggesting the possible absence of GME. 
After distillation, the fidelity increases to $0.572(23)$, surpassing the GME threshold of $0.5$ by more than three standard deviations, providing experimental evidence for GME superactivation. 
To establish this conclusion rigorously, it is crucial to ensure the absence of GME in initial states, as noise in state preparation can significantly weaken the effectiveness of the entanglement witness~\cite{liu2022fundamental}. 
For this purpose, we perform state tomography on the two initial states.
Using methods developed in Refs.~\cite{ma2011gme,hashemi2012genuine,peres1996ppt,jungnitsch2011taming,ties2024sep} and the tomographic data, we numerically verify the absence of GME in initial states and confirm the occurrence of GME superactivation. 
As for the states with target parameters of $p=0.4$ and $p=0.32$, although clear improvements in fidelities are also observed after distillation, the fidelities do not exceed $0.5$.
This suggests that GME superactivation is confined in a narrow parameter range.

As another entanglement resource for multipartite quantum states, some states in our experiment possess SLE despite the absence of GME, such as the two initial states with $p=0.5$. 
This implies that we can localize their entanglement into specific subsystems and extract EPR pairs. 
For an ideal three-photon noisy GHZ state, measuring one photon in Pauli-$X$ basis is one of the optimal entanglement localization operations. 
In our experiment, by projecting Bob's photon to the state $\ket{+}$ or $\ket{-}$, the fidelities between the remaining two-photon states and the EPR pair both exceed the bipartite entanglement threshold of $0.5$.
This demonstrates a successful entanglement localization, as indicated by the two triangles for $p=0.5$ in Fig.~\ref{Fig:data}(b).

We find that entanglement localization can benefit from tripartite distillation. 
As indicated by two triangles at $p=0.4$ in Fig.~\ref{Fig:data}(b), the localization operation fails to extract EPR pairs from both initial states.
However, by pre-executing the tripartite distillation operation on two initial states before the localization step, the fidelity between the final two-photon state and the EPR pair exceeds $0.5$. This process is indicated by the green and blue arrows. 
Thus, we experimentally observe that tripartite distillation facilitates the ability to extract EPR pairs.
Note that during the processes of entanglement distillation and localization, we effectively perform a Bell state measurement on Bob's two photons, which enables the execution of other multipartite entanglement manipulation protocols~\cite{PurificationswappingPRA1999}.
However, due to the deviation between the experimentally prepared state and the ideal one,  the failure of a specific localization operation based on Pauli-$X$ basis measurement cannot rule out the existence of SLE in initial states of $p=0.4$.
We thus perform tomography on two initial states with $p=0.36$ and confirm the absence of SLE with the SLE existence criterion in Section II.A of Supplementary Material~\cite{Supplemental}.
After applying distillation and localization operations sequentially to the two initial states, we observe the fidelity between the remaining two-photon state and EPR pair exceeds $0.5$. 
Since SLOCC cannot create SLE, this result experimentally confirms our theoretical prediction of SLE superactivation.

\noindent \emph{Discussion.---}
In this work, by constructing a tripartite entanglement distillation network in an eight-photon linear optical platform, we demonstrate that entanglement superactivation is a versatile tool for resource recycling in quantum networks. 
Our work accomplishes these key experimental advancements: the first experimental demonstration of GME superactivation and the discovery of SLE superactivation.
Despite these advancements, the probabilistic nature of SPDC sources imposes rate limitations, future integration of deterministic sources (e.g., quantum dots~\cite{uppu2020scalable}) could enable higher-efficiency state preparation.
Furthermore, the development of entanglement distillation techniques based on CNOT gates could facilitate more general distillation operations~\cite{bennett1996purification,MuraoPRA1998}.

Our main theoretical contributions focus on SLE, including its definition and existence criterion, as well as SLE superactivation, encompassing its prediction and comparison with GME superactivation. 
As a multipartite entanglement concept with profound physical significance, we believe it is important to identify the necessary and sufficient conditions for SLE superactivation and to explore its connections with other notions such as state interconvertibility~\cite{dur2000three}. Pursuing these directions may also provide insights into longstanding open problems, including the PPT square conjecture~\cite{ruskai2012operator}.

\noindent \emph{Note added.---} In revising the manuscript, we noted a similar GME superactivation experiment in the trapped-ion system~\cite{starek2025experimental}.

\noindent \emph{Acknowledgements.---} We appreciate insightful discussions with Ties-A. Ohst, Qingyu Li, Jens Eisert, Satoya Imai, Hayata Yamasaki, Xiaodong Yu, Luo-Kan Chen, and Feihu Xu.
This work was supported by the National Natural Science Foundation of China (Grants No.~11975222, No.~11874340, No.~12174216), Shanghai Municipal Science and Technology Major Project (Grant No.~2019SHZDZX01), Chinese Academy of Sciences and the Shanghai Science and Technology Development Funds (Grant No.~18JC1414700), and the Innovation Program for Quantum Science and Technology (Grant No.~2021ZD0301901, No.~2021ZD0300804).
Xu-Fei Yin was supported from the China Postdoctoral Science Foundation (Grant No.~ 2023M733418).
Otfried Gühne was supported by the Deutsche Forschungsgemeinschaft (DFG, German Research Foundation, project number 563437167), the 
Sino-German Center for Research Promotion (Project M-0294), and the
German Federal Ministry of Research, Technology and Space (Project QuKuK, Grant No.~16KIS1618K and Project BeRyQC, Grant No.~13N17292).
Yu-Ao Chen was supported by the XPLORER PRIZE from Tencent Foundation.

%

\end{document}


\title{Supplementary Material for ``Entanglement Superactivation in Multiphoton Distillation Networks"}

\maketitle

\tableofcontents
\clearpage


\section{Background}\label{sec:preliminary}
We first give some basic definitions which will be frequently used in this work. An $N$-partite state $\rho$ is said to be bi-separable if it can be written as
\begin{equation}\label{eq:bi_sep}
\rho=\sum_{g\subset[N]}\sum_ip_{gi}\rho_{g}^i\otimes\rho_{\Bar{g}}^i,
\end{equation}
where $g$ denotes a nontrivial subset of the $N$ parties, $\Bar{g}=[N]-g$ is the complementary of $g$, $\{p_{gi}\}_{g,i}$ is a probability distribution satisfying $\sum_{g,i}p_{g,i}=1$ and $p_{gi}\ge 0$, $\rho_{g}^i$ and $\rho_{\Bar{g}}^i$ are quantum states defined on systems $g$ and $\Bar{g}$. If an $N$-partite state $\rho$ cannot be decomposed into this form, it is a genuinely multipartite entangled state. It is worth mentioning that a bi-separable state may also have quantum resources. For example, it can be entangled with respect to a given bi-partition.

Local operations and classical communication (LOCC) can be written in a separable form \cite{Chitambar2014LOCC} 
\begin{equation}\label{eq:locc}
\Lambda(\rho)=\sum_i\left(K_1^i\otimes\cdots\otimes K_N^i\right)\rho\left(K_1^i\otimes\cdots\otimes K_N^i\right)^\dagger,
\end{equation}
where $\rho$ is an $N$-partite quantum state and $K$ denotes the Kraus operator satisfying 
\begin{equation}\label{eq:normalization}
\sum_i\left(K_1^i\otimes\cdots \otimes K_N^i\right)^{\dagger}\left(K_1^i\otimes\cdots \otimes K_N^i\right)=\mathbb{I}.
\end{equation}
It is worth mentioning, however, that not all separable maps
as in Eq.~(\ref{eq:locc}) can be implemented via LOCC~\cite{bennett1999nonlocal}.
Stochastic local operations and classical communication (SLOCC) are defined based on LOCC, which permits post-selection of the resultant states in Eq.~\eqref{eq:locc}. Thus, SLOCC can also be written as the separable form of Eq.~\eqref{eq:locc} with 
\begin{equation}
\sum_i\left(K_1^i\otimes\cdots \otimes K_N^i\right)^\dagger\left(K_1^i\otimes\cdots \otimes K_N^i\right)\le\mathbb{I}.
\end{equation}
Due to the stochastic property, a SLOCC operation may succeed with a probability less than $1$.

It can be proved that neither LOCC nor SLOCC can distill GME from one copy of a bi-separable state. As for bi-separable state $\rho$,
\begin{equation}
\Lambda(\rho)=\sum_{g\in[N]}\sum_ip_{gi}\sum_j\left(K_1^j\otimes\cdots\otimes K_N^j\right)\rho_g^i\otimes\rho_{\Bar{g}}^i\left(K_1^j\otimes\cdots\otimes K_N^j\right)^\dagger
\end{equation}
is also a bi-separable state with a form similar to Eq.~\eqref{eq:bi_sep}. 

It has been shown that the simultaneous preparation of multiple bi-separable states can activate GME \cite{Yamasaki2022activation}. 
For example, if we collect two bi-separable tripartite states $\sigma_{ABC}$ and $\sigma^\prime_{A^\prime B^\prime C^\prime}$ and regard them as a whole to be a new tripartite state $\rho_{AA^\prime,BB^\prime,CC^\prime}=\sigma_{ABC}\otimes\sigma^\prime_{A^\prime B^\prime C^\prime}$, this new state can have tripartite GME over the partition $AA^\prime$, $BB^\prime$, and $CC^\prime$.

\section{GME and SLE verification}\label{sec:verification_methods}

\subsection{Theoretical criteria}
As the theme of our work is entanglement superactivation, including GME and SLE superactivation, it is necessary to ensure that to-be-distilled states have no GME and SLE. 
In the main text, we adopt two fidelity-based entanglement witnesses to indicate the absence of these two properties, which are not perfect. To further ensure the absence, we need more rigorous GME and SLE verification methods, which will be discussed in this section.

For a general multipartite state, verifying the existence of GME is challenging, and verifying the absence of GME is even more challenging. In theory, necessary and sufficient conditions of GME exist for some density matrices with special forms. For example, every ``X''-shaped density matrix is entangled if and only if its GME concurrence is larger than 0 \cite{ma2011gme,hashemi2012genuine}. Here, a density matrix is said to be of ``X''-shaped if we can write it as a form of
\begin{equation}
\rho=
\begin{bmatrix}
M_1 & M_2 \\
M_2^\dagger & M_3
\end{bmatrix},
\end{equation}
where $M_1$, $M_2$, and $M_3$ are square matrices satisfying $M_1=\mathrm{diag}(a_1,\cdots,a_L)$, $M_3=\mathrm{diag}(b_L,\cdots,b_1)$, and $M_2=\mathrm{antidiag}(c_1,\cdots,c_L)$.
The GME concurrence for this state is
\begin{equation}\label{eq:GME_concurrence}
C_{\mathrm{GME}}\left(\begin{bmatrix}
M_1 & M_2 \\
M_2^\dagger & M_3
\end{bmatrix}\right)=2\max_i\{0,|c_i|-\sum_{j\neq i}\sqrt{a_jb_j}\}.
\end{equation}

We introduce the GME concurrence method for ``X''-shaped density matrices because our target states, including the noisy GHZ state, the distilled noisy GHZ state, and the noisy GHZ state with bit-flip and phase-flip errors, all have ``X''-shaped density matrices. 
However, due to unpredictable noises, real density matrices may deviate from the ideal ``X''-shaped matrices. 
So here, we also use another GME verification method for states with general forms of density matrices, the so-called PPT mixer approach. If the solution of the following semidefinite programming problem for a tripartite state $\rho_{ABC}$ is negative, then $\rho_{ABC}$ is genuinely multipartite entangled \cite{jungnitsch2011taming},
\begin{equation}\label{eq:PPT_optimization}
\begin{aligned}
    \min_W & \tr(W\rho_{ABC})\\
    \mathrm{s.t.} & \tr(W)=1,\\
    & W=P_1+Q_1^{\mathrm{T}_A} \ , \ P_1\ge 0 \ , \ Q_1\ge 0,\\
    & W=P_2+Q_2^{\mathrm{T}_B} \ , \ P_2\ge 0 \ , \ Q_2\ge 0,\\
    & W=P_3+Q_3^{\mathrm{T}_C} \ , \ P_3\ge 0 \ , \ Q_3\ge 0,\\
\end{aligned}
\end{equation}
where $\mathrm{T}_A$ denotes the partial transposition operation for indices belonging to subsystem $A$. This criterion is generalized from the PPT criterion \cite{peres1996ppt} for bipartite entanglement verification. 
It is worth noting that most entanglement criteria, including this, are sufficient but not necessary conditions. 
If the solution of this optimization problem is positive for some state $\rho_{ABC}$, we cannot be sure that this state is bi-separable. 
Nonetheless, as PPT is a powerful entanglement criterion (in fact, it is necessary and sufficient for several families of states, including X-shaped states), the positive value of this optimization problem provides strong evidence for bi-separability.
In this work, we will use this PPT-based criterion as evidence for GME of a given tripartite state.

Let us now discuss the theoretical characterization of SLE. First, we have the following criterion.
\begin{theorem}[SLE Verification Criterion]\label{theorem:main}
	An $N$-partite state $\rho$ possesses SLE on subsystems $A$ and $B$ if and only if there exists a pure tensor state $\ket{\psi}=\bigotimes_{g\in \overline{AB}}\ket{\psi_g}$ defined on the complementary set $\overline{AB}$, such that $\rho_{AB}=\frac{\bra{\psi}\rho\ket{\psi}}{\tr(\bra{\psi}\rho\ket{\psi})}$ is an entangled state.
\end{theorem}
\begin{proof}
	If $\rho_{AB}$ defined above is entangled, one can perform POVM measurements $\{\ketbra{\psi_g},\mathbb{I}-\ketbra{\psi_g}\}_g$ on all the other subsystems.
	After measurement, $\rho_{AB}$ is kept only when other subsystems get results of $\{\ketbra{\psi_g}\}_g$.
	Through this SLOCC protocol, one successfully produces $\rho_{AB}=\frac{\bra{\psi}\rho\ket{\psi}}{\tr(\bra{\psi}\rho\ket{\psi})}$, where $\ket{\psi}=\bigotimes_{g\in \overline{AB}}\ket{\psi_g}$, from $\rho$ and the entanglement of $\rho_{AB}$ shows that $\rho$ indeed has SLE.
	
	By definition, if $\rho$ has SLE on subsystems $A$ and $B$, there exists an SLOCC operation with separable Kraus operators $\{K^i=K_A^i\otimes K_B^i\otimes\cdots\}_i$ such that
		\begin{eqnarray} 
		\frac{\tr_{\overline{AB}}(\sum_i K^i\rho K^{i\dagger})}{\tr(\sum_j K^j\rho K^{j\dagger})}=\sum_i\frac{\tr_{\overline{AB}}( K^i\rho K^{i\dagger})}{\tr(\sum_j K^j\rho K^{j\dagger})}
		\end{eqnarray} 
	is an entangled state on $A$ and $B$. As the set of separable states is convex, there exists at least one Kraus operator $K= \bigotimes_{g\in[N]}K_g$, such that
		\begin{eqnarray} 
		\frac{\tr_{\overline{AB}}\left(K\rho K^\dagger\right)}{\tr\left(K\rho K^\dagger\right)}=\frac{(K_A\otimes K_B)\tr_{\overline{AB}}\left[\left(\bigotimes_{g\in \overline{AB}}K_g\right)\rho \left(\bigotimes_{g\in \overline{AB}}K_g^\dagger\right)\right](K_A\otimes K_B)^\dagger}{\tr\left(K\rho K^\dagger\right)}
		\end{eqnarray} 
	is an entangled state. As SLOCC cannot activate bipartite entanglement, this means that\\ $\tr_{\overline{AB}}\left[\left(\bigotimes_{g\in \overline{AB}}K_g\right)\rho \left(\bigotimes_{g\in \overline{AB}}K_g^\dagger\right)\right]$ is an unnormalized entangled state. According to the singular value decomposition, we can decompose $K_g$ as $K_g=\sum_i\sqrt{\lambda_g^{i_g}} \ketbra{v_{g}^{i_g}}{u_{g}^{i_g}}$ where $\sqrt{\lambda_g^{i_g}}\ge0$ is the singular value of $K_g$, $\{\ket{u_{i_g}}\}_i$ and $\{\ket{v_{i_g}}\}_i$ are sets of mutually orthogonal states. 
	Substituting the singular value decomposition, we have
		\begin{equation}
		\label{eq:singular_state}
		\begin{aligned}
		&\tr_{\overline{AB}}\left[\left(\bigotimes_{g\in \overline{AB}}K_g\right)\rho \left(\bigotimes_{g\in \overline{AB}}K_g^\dagger\right)\right]\\
		=&\left(\sum_{g\in\overline{AB}}\sum_{i_g,j_g}\right)\left(\prod_{g\in\overline{AB}}\sqrt{\lambda_g^{i_g}\lambda_g^{j_g}}\braket{v_g^{j_g}}{v_g^{i_g}}\right)\left(\bigotimes_{g\in\overline{AB}}\bra{u_g^{i_g}}\right)\rho\left(\bigotimes_{g\in\overline{AB}}\ket{u_g^{j_g}}\right)\\
		=&\left(\sum_{g\in\overline{AB}}\sum_{i_g}\right)\left(\prod_{g\in\overline{AB}}\lambda_g^{i_g}\right)\left(\bigotimes_{g\in\overline{AB}}\bra{u_g^{i_g}}\right)\rho\left(\bigotimes_{g\in\overline{AB}}\ket{u_g^{i_g}}\right),
		\end{aligned}
		\end{equation}
	where we use the property of $\braket{v_g^{j_g}}{v_g^{i_g}}=\delta_{i,j}$. As $\prod_{g\in\overline{AB}}\lambda_g^{i_g}\ge0$, if the state in Eq.~\eqref{eq:singular_state} is entangled, at least one element in the summation, $\left(\bigotimes_{g\in\overline{AB}}\bra{u_g^{i_g}}\right)\rho\left(\bigotimes_{g\in\overline{AB}}\ket{u_g^{i_g}}\right)$, is entangled. This concludes our proof. Note that similar arguments appeared also in footnote [18] of Ref.~\cite{gunhe2016emergent}
	
\end{proof}

Although SLE seems to represent a lower level of multipartite entanglement, Ref.~\cite{gunhe2016emergent} found a genuine multipartite entangled state, which has no SLE according to this theorem, showing the independence of SLE and GME. 

The certification of SLE can be reduced to a bipartite entanglement problem based on this theorem. 
Specifically, to certify if a tripartite state $\rho_{ABC}$ has SLE on subsystems $A$ and $B$, we can use the following function
\begin{equation}\label{eq:SLE_calculation}
\mathcal{E}_{\mathrm{SL}}(\rho_{ABC})=\max_{\psi_C}\mathcal{E}\left(\frac{\bra{\psi_C}\rho_{ABC}\ket{\psi_C}}{\tr(\bra{\psi_C}\rho_{ABC}\ket{\psi_C})}\right),
\end{equation}
where $\mathcal{E}(\cdot)$ is a bipartite entanglement quantifier. 
If $A$ and $B$ are two qubits, we can assign $\mathcal{E}(\cdot)$ to be the entanglement negativity $\mathcal{N}(\rho_{AB})=\log(\tr|\rho_{AB}^{\mathrm{T}_B}|)$, which is defined by the violation of PPT criterion \cite{peres1996ppt}. 
In this case, $\rho_{ABC}$ has SLE on $A$ and $B$ if and only if $\mathcal{E}_{\mathrm{SL}}(\rho_{ABC})>0$, as PPT criterion is the necessary and sufficient condition for two-qubit entanglement.

\subsection{Noisy initial states} \label{sec:SLE_GME_noisy}
To show that we realize the GME and SLE superactivation in our experiments, we need to prove the absence of GME in to-be-distilled noisy states of $p=0.5$ and the absence of SLE in to-be-distilled noisy states of $p=0.36$, as shown in Fig.~3 of the main text. 
Due to unavoidable experimental noise, GHZ states with ideal white noise cannot be prepared.
We mainly adopt two methods to describe these noisy states. The first is based on our knowledge of errors in experiments. 
In preparing the GHZ state, we know that dominant errors are phase-flip and bit-flip, which can convert the ideal GHZ state $\ket{G_0^+}$ into seven other GHZ states. 
All these eight GHZ states are listed as following
\begin{equation} \label{eq:eight_GHZ}
\begin{aligned}
\ket{G_0^{\pm}}&=\frac{1}{\sqrt{2}}(\ket{000}\pm\ket{111}),\\
\ket{G_1^{\pm}}&=\frac{1}{\sqrt{2}}(\ket{001}\pm\ket{110}),\\
\ket{G_2^{\pm}}&=\frac{1}{\sqrt{2}}(\ket{010}\pm\ket{101}),\\
\ket{G_3^{\pm}}&=\frac{1}{\sqrt{2}}(\ket{100}\pm\ket{011}).
\end{aligned}
\end{equation}
Because of the symmetry in light paths, we assume three photons have the same error rate.
Thus, we can use three parameters to describe the to-be-distilled states that we actually prepare
\begin{equation}\label{eq:real_state}
\rho = p\left[r\rho_0+\frac{1-r}{3}\left(\rho_1+\rho_2+\rho_3\right)\right]+(1-p)\frac{\mathbb{I}_3}{8},
\end{equation}
where $\rho_i = q\ketbra{G_i^+}{G_i^+}+(1-q)\ketbra{G_i^-}{G_i^-}$, $p$, and $q$, and $r$ are parameters representing rates of white noise, phase-flip error, and bit-flip error, respectively.  

In addition to the noise model, we performed quantum state tomography measurement on to-be-distilled states. Thus, we can use reconstructed density matrices shown in Sec.~\ref{sec:tomo} to analyze the GME and SLE. 
It is worth mentioning that, due to experimental errors in quantum tomography, the reconstructed density matrix may not be closer to the real state than the noise model, Eq.~\eqref{eq:real_state}. 
Therefore, it is meaningful to adopt both methods to analyze the existence of GME and SLE.

We start from the GME analysis of two to-be-distilled states of $p=0.5$. Firstly, with the noise model in Eq.~\eqref{eq:real_state}, the density matrix of $\rho$ is
\begin{equation}\label{eq:real_state_matrix}
\rho=
\begin{bmatrix}
\frac{1-p}{8}+\frac{1}{2}pr & 0 & 0 & 0 & 0 & 0 & 0 & pr(q-\frac{1}{2})\\
0 & \frac{1}{8}+\frac{p-4pr}{24} & 0 & 0 & 0 & 0 & \frac{1-r}{3}p(q-\frac{1}{2}) & 0\\
0 & 0 & \frac{1}{8}+\frac{p-4pr}{24} & 0 & 0 & \frac{1-r}{3}p(q-\frac{1}{2}) & 0 & 0\\
0 & 0 & 0 & \frac{1}{8}+\frac{p-4pr}{24} & \frac{1-r}{3}p(q-\frac{1}{2}) & 0 & 0 & 0\\
0 & 0 & 0 & \frac{1-r}{3}p(q-\frac{1}{2}) & \frac{1}{8}+\frac{p-4pr}{24} & 0 & 0 & 0\\
0 & 0 & \frac{1-r}{3}p(q-\frac{1}{2}) & 0 & 0 & \frac{1}{8}+\frac{p-4pr}{24} & 0 & 0\\
0 & \frac{1-r}{3}p(q-\frac{1}{2}) & 0 & 0 & 0 & 0 & \frac{1}{8}+\frac{p-4pr}{24} & 0\\
pr(q-\frac{1}{2}) & 0 & 0 & 0 & 0 & 0 & 0 & \frac{1-p}{8}+\frac{1}{2}pr
\end{bmatrix},
\end{equation}
which is clearly an ``X"-shaped density matrix. Thus, to decide whether it is genuinely entangled or not, we only need to calculate the GME concurrence, 
\begin{equation}\label{eq:concurrence_noise_model}
C_{\mathrm{GME}}\left(\rho\right)=2\max\{0,pr(q-\frac{1}{2})-\frac{3}{8}-\frac{p-4pr}{8},\frac{1}{3}p(1-r)(q-\frac{1}{2})-\frac{1-p}{8}-\frac{pr}{2}-\frac{1}{4}-\frac{p-4pr}{12}\}.
\end{equation}
Using the reconstructed density matrix of Fig.~\ref{Fig:tomo_1.0}, we find that the parameters of two initial states are approximately $(p,q,r)=(0.5,0.9084,0.9210)$ and $(0.5,0.9124,0.9129)$. Substituting them into Eq.~\eqref{eq:concurrence_noise_model}, GME concurrences are $2\max\{0,-0.0192,-0.4255\}=0$ and $2\max\{0,-0.0210,-0.4243\}=0$, which show that these two states have no GME.

We can also directly use the reconstructed density matrices from quantum tomography to analyze the GME. 
As reconstructed density matrices are not strict ``X"-shaped matrices, calculating GME concurrence is challenging \cite{ma2011gme}. We thus use the PPT-based GME verification method in Eq.~\eqref{eq:PPT_optimization} to test these two states. 
We input reconstructed density matrices of the two to-be-distilled states, shown in Fig.~\ref{Fig:tomo_0.5}, into this optimization problem and solve it numerically. 
Results are $0.0096$ and $0.0066$ for two density matrices. 
These two values are all greater than 0, so this PPT-based verification method cannot detect their GME. 
As this method does not ask for an ``X''-shaped density matrix and the PPT criterion is a strong entanglement criterion, these results provide strong evidence that to-be-distilled states have no GME.

In addition, we also numerically verify it.
We adopt the code from Ref.~\cite{ties2024sep} to show that these two states all have the bi-separable decomposition in Eq.~\eqref{eq:bi_sep}, which provides a sufficient condition for the absence of GME.

We now analyze the SLE of to-be-distilled states with $p=0.36$ based on Theorem~\ref{theorem:main}. 
Here, we numerically verify SLE by searching over all single-qubit pure states $\ket{\psi_C}=\cos\theta\ket{0}+\sin\theta e^{i\phi}\ket{1}$ to find a positive value of $\mathcal{E}_{\mathrm{SL}}(\rho_{ABC})$ defined in Eq.~\eqref{eq:SLE_calculation}, where $\mathcal{E}(\cdot)$ is chosen to be entanglement negativity. 
We numerically calculate that the noisy state in Eq.~\eqref{eq:real_state} with parameters $(p,q,r)=(p,0.9084,0.9210)$ has no SLE when $p<0.4095$ and the state with parameters $(p,q,r)=(p,0.9124,0.9129)$ has no SLE when $p<0.4103$. 
This shows that two to-be-distilled states with $p=0.36$ have no SLE. 
We also directly use reconstructed density matrices shown in Fig.~\ref{Fig:tomo_0.36} to prove the absence of SLE. 
After calculation, values of $\mathcal{E}_{\mathrm{SL}}(\rho_{ABC})$ of these two to-be-distilled states are both 0. 
To further consolidate this conclusion, we change $\mathcal{E}(\cdot)$ to be the lowest eigenvalue of the input bipartite density matrix after partial transposition, $\mathcal{E}(\rho_{AB})=\lambda_{\min}(\rho_{AB}^{\mathrm{T}_B})$.
In this case, if $\mathcal{E}_{\mathrm{SL}}(\rho_{ABC})$ is larger than $0$, $\rho_{ABC}$ has no SLE.
We numerically calculate the values of $\mathcal{E}_{\mathrm{SL}}(\rho_{ABC})$ for two to-be-distilled states to be $0.036$ and $0.029$, which confirms the absence of SLE in these two states.

\section{GME and SLE superactivation} \label{sec:ideal_Werner}

\subsection{Noisy GHZ state}\label{subsec:noisy_GHZ}
The noisy GHZ state has a simple mathematical form. We can use it as a theoretical example to analyze the superactivation of GME and SLE. 
The noisy GHZ state is parameterized by a single parameter $p$ as 
\begin{equation}\label{eq:ideal_werner}
    \rho_{\mathrm{noise}}=p\ketbra{G_0^+}{G_0^+}+\frac{1-p}{8}\mathbb{I}_3,
\end{equation}
which is apparently an ``X"-shaped density matrix, with GME concurrence
\begin{equation}
C_{\mathrm{GME}}(\rho_{\mathrm{noise}})=2\max\{0,\frac{7}{8}p-\frac{3}{8}\}.
\end{equation}
Thus, the noisy GHZ state is genuinely multipartite entangled if and only if $p>\frac{3}{7}$. After the distillation procedure shown in Sec.~\ref{sec:distillation_protocol}, the noisy GHZ state  becomes
\begin{equation}\label{eq:dis_werner}
\rho_{\mathrm{noise}}^\prime=\frac{1}{3p^2+1}
\begin{bmatrix}
\frac{(3p+1)^2}{8} & 0 & 0 & 0 & 0 & 0 & 0 & 2p^2\\
0 & \frac{(1-p)^2}{8} & 0 & 0 & 0 & 0 & 0 & 0\\
0 & 0 & \frac{(1-p)^2}{8} & 0 & 0 & 0 & 0 & 0\\
0 & 0 & 0 & \frac{(1-p)^2}{8} & 0 & 0 & 0 & 0\\
0 & 0 & 0 & 0 & \frac{(1-p)^2}{8} & 0 & 0 & 0\\
0 & 0 & 0 & 0 & 0 & \frac{(1-p)^2}{8} & 0 & 0\\
0 & 0 & 0 & 0 & 0 & 0 & \frac{(1-p)^2}{8} & 0\\
2p^2 & 0 & 0 & 0 & 0 & 0 & 0 & \frac{(3p+1)^2}{8}
\end{bmatrix},
\end{equation}
which is also an ``X''-shaped density matrix. Thus, the GME concurrence of $\rho^\prime_{\mathrm{noise}}$ can also be calculated as
\begin{equation}
C_{\mathrm{GME}}(\rho^\prime_{\mathrm{noise}})=\frac{1}{3p^2+1}\max\{0,2p^2-\frac{3}{8}(1-p)^2\}.
\end{equation}
To let $C_{\mathrm{GME}}(\rho^\prime_{\mathrm{noise}})>0$, we can find that $p>\frac{4\sqrt{3}-3}{13}\sim0.3022$. Therefore, in the range of $p\in(\frac{4\sqrt{3}-3}{13},\frac{3}{7})$, the noisy GHZ state has the phenomenon of GME superactivation by collecting two copies. 
Note that $\frac{4\sqrt{3}-3}{13}$ may not be the lowest value of GME superactivation with two copies of noisy GHZ states, as we derive it with a specific distillation protocol.

Now, we turn to analyze SLE. According to Eq.~\eqref{eq:SLE_calculation}, the calculation of SLE on subsystems $A$ and $B$ needs to search over all pure states on subsystem $C$, $\ket{\psi_C}=\cos{\theta}\ket{0}+\sin{\theta}e^{i\phi}\ket{1}$. For the noisy GHZ state,
\begin{equation}
\frac{\bra{\psi_C}\rho_{\mathrm{noise}}\ket{\psi_C}}{\tr(\bra{\psi_C}\rho_{\mathrm{noise}}\ket{\psi_C})}=p(\cos{\theta}\ket{00}+\sin{\theta}e^{-i\phi}\ket{11})(\cos{\theta}\bra{00}+\sin{\theta}e^{i\phi}\bra{11})+\frac{1-p}{4}\mathbb{I}_2.
\end{equation}
It is easy to check that when searching over all the possible values of $\theta$ and $\phi$, the maximal entanglement is obtained when $\theta=\frac{\pi}{4}$, with the state $\frac{p}{2}(\ket{00}+e^{-i\phi}\ket{11})(\bra{00}+e^{i\phi}\bra{11})+\frac{1-p}{4}\mathbb{I}_2$. It can be verified using the PPT criterion that when $p>\frac{1}{3}$, this two-qubit state is entangled. 
Therefore, the noisy GHZ state has SLE if and only if $p>\frac{1}{3}$.

According to Eq.~\eqref{eq:dis_werner}, after distillation, the noisy GHZ state can be written as
\begin{equation}
\rho_{\mathrm{noise}}^\prime=\frac{3p^2+p}{3p^2+1}\ketbra{G_0^+}{G_0^+}+\frac{p-p^2}{3p^2+1}\ketbra{G_0^-}{G_0^-}+\frac{(1-p)^2}{8(3p^2+1)}\mathbb{I}_3.
\end{equation}
Thus, after the inner product with the reference state, we have
\begin{equation}\label{eq:inner_product}
\frac{\bra{\psi_C}\rho^\prime_{\mathrm{noise}}\ket{\psi_C}}{\tr(\bra{\psi_C}\rho^\prime_{\mathrm{noise}}\ket{\psi_C})}=\frac{3p^2+p}{3p^2+1}
\begin{bmatrix}
    \cos^2{\theta} & 0 & 0 & \frac{1}{2}\sin{2\theta}e^{i\phi}\\
    0 & 0 & 0 & 0 \\
    0 & 0 & 0 & 0 \\
    \frac{1}{2}\sin{2\theta}e^{-i\phi} & 0 & 0 & \sin^2{\theta}
\end{bmatrix}
+ \frac{p-p^2}{3p^2+1}
\begin{bmatrix}
    \cos^2{\theta} & 0 & 0 & -\frac{1}{2}\sin{2\theta}e^{i\phi}\\
    0 & 0 & 0 & 0 \\
    0 & 0 & 0 & 0 \\
    -\frac{1}{2}\sin{2\theta}e^{-i\phi} & 0 & 0 & \sin^2{\theta}
\end{bmatrix}
+ \frac{(1-p)^2}{4(3p^2+1)}\mathbb{I}_2.
\end{equation}
To decide whether this state is entangled or not, we can take a partial transposition of this density matrix and calculate its lowest eigenvalue. The partial-transposed matrix is
\begin{equation}
\left[\frac{\bra{\psi_C}\rho^\prime_{\mathrm{noise}}\ket{\psi_C}}{\tr(\bra{\psi_C}\rho^\prime_{\mathrm{noise}}\ket{\psi_C})}\right]^{\mathrm{T}_B}=
\begin{bmatrix}
    \frac{2(p^2+p)}{3p^2+1}\cos^2{\theta} & 0 & 0 & 0\\
    0 & 0 & \frac{2p^2}{3p^2+1}\sin{2\theta}e^{-i\phi} & 0 \\
    0 & \frac{2p^2}{3p^2+1}\sin{2\theta}e^{i\phi} & 0 & 0 \\
    0 & 0 & 0 & \frac{2(p^2+p)}{3p^2+1}\sin^2{\theta}
\end{bmatrix}
+ \frac{(1-p)^2}{4(3p^2+1)}\mathbb{I}_2,
\end{equation}
whose lowest eigenvalue is 
\begin{equation}
    \lambda_{\min}\left\{\left[\frac{\bra{\psi_C}\rho^\prime_{\mathrm{noise}}\ket{\psi_C}}{\tr(\bra{\psi_C}\rho^\prime_{\mathrm{noise}}\ket{\psi_C})}\right]^{\mathrm{T}_B}\right\}=\frac{(1-p)^2}{4(3p^2+1)}-\frac{2p^2}{3p^2+1}\sin{2\theta}\ge\frac{(1-p)^2}{4(3p^2+1)}-\frac{2p^2}{3p^2+1}.
\end{equation}
If we require the lowest eigenvalue to be lower than 0, we have $p>\frac{2\sqrt{2}-1}{7}\sim0.2612$. So in the region of $p\in (\frac{2\sqrt{2}-1}{7},\frac{1}{3})$, two Werner states can be collected to activate SLE based on our distillation protocol. Similarly, $p=\frac{2\sqrt{2}-1}{7}$ may not be the lowest value for SLE superactivation with two copies, as we only consider a special distillation scheme here.

The overall fidelity is also a crucial indicator for our entanglement distillation and localization protocols.
With the analysis in this section, we can derive analytical forms.
Using Eq.~\eqref{eq:dis_werner}, the fidelity between the distilled state and the GHZ state is $F_1=\frac{25p^2+6p+1}{8(3p^2+1)}$.
Using Eq.~\eqref{eq:inner_product}, the fidelity between the distilled and localized state and the Bell state $\left(\ket{00}+\ket{11}\right)/\sqrt{2}$ is $F_2=\frac{13p^2+2p+1}{4(3p^2+1)}$. 

To clarify the conclusions, we use Fig.~\ref{Fig:ent_level} to summarize ranges for different properties of the noisy GHZ state. In this diagram, we use $p_{\mathrm{GME}}^{(2)}$ and $p_{\mathrm{SLE}}^{(2)}$ to label the thresholds for using two copies of noisy GHZ states to active GME and SLE with our distillation protocol. Since it is only a special distillation protocol, these two values are upper bounds for entanglement superactivation thresholds with two copies. A meaningful problem is the derivation of the exact thresholds for SLE and GME superactivation with a finite number of copies. Besides, it has been proved that a multipartite state can always be used to activate GME given enough copies if and only if it is not separable in any fixed bi-partition\cite{Palazuelos2022genuinemultipartite}. This means that $p_{\mathrm{GME}}^{(\infty)}=0.2$ for the noisy GHZ state. Consequently, another important problem is that what is $p_{\mathrm{SLE}}^{(\infty)}$? In what condition can SLE not be activated even given an arbitrary number of the multipartite state? It is evident that a fully separable state has no SLE, so $0.2$ is a lower bound for $p_{\mathrm{SLE}}^{(\infty)}$.

\begin{figure}[htbp!]
\centering
\includegraphics[width=0.7 \linewidth]{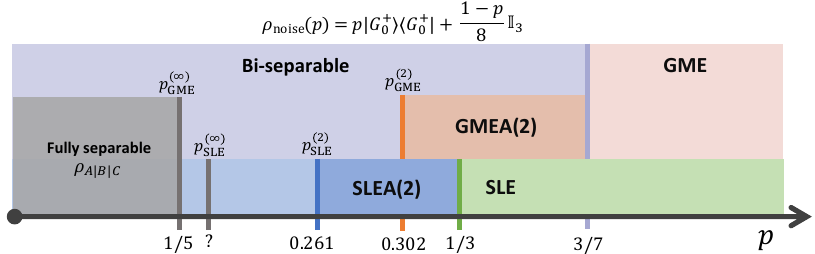}
\caption{Ranges of different properties for the noisy GHZ state.}
\label{Fig:ent_level}
\end{figure}

\subsection{Noisy W State} \label{sec:W_activation}
Till now, we only focus on noisy GHZ state, the noisy GHZ state. 
In this section, we would like to show if the noisy W state has the entanglement superactivation phenomenon,
\begin{equation}
\rho_{\mathrm{W}}=p\ketbra{\mathrm{W}}{\mathrm{W}}+\frac{1-p}{8}\mathbb{I}_3,
\end{equation}
where $\ket{\mathrm{W}}=\frac{1}{\sqrt{3}}\left(\ket{001}+\ket{010}+\ket{100}\right)$.
It is known that, all three-qubit genuinely entangled state can be classified into either GHZ or W classes.
Thus, it is meaningful to discuss the superactivation phenomenon in both noisy GHZ and W states.

We here consider the CNOT gate-based entanglement distillation protocol, which is shown in Fig.~\ref{fig:W_state}.
Take the two-qubit case as an example, each qubit of $\rho$ is transmitted to a client.
Then, each client perform a CNOT gate on its two qubit, and measure one of them in computational basis.
When two measurement results are all $\ket{0}$, the other two qubits will be kept.
If some client get the measurement result of $\ket{1}$, the other two qubits will be discarded.
Note that $\mathrm{CNOT}=\ketbra{0}{0}\otimes\mathbb{I}+\ketbra{1}{1}\otimes X$, after post-selection, the operator becomes 
\begin{equation}
P_0=(\mathbb{I}\otimes\bra{0})\mathrm{CNOT}=\ketbra{0}{0}\otimes\bra{0}\mathbb{I}+\ketbra{1}{1}\otimes \bra{0}X=\ketbra{0}{00}+\ketbra{1}{11}.
\end{equation}
Thus, after this distillation process, the post-selected state becomes $\frac{P_0^{\otimes 2}\rho^{\otimes 2}P_0^{\dagger\otimes 2}}{\Tr(P_0^{\otimes 2}\rho^{\otimes 2}P_0^{\dagger\otimes 2})}$.
Mathematically speaking, the matrix element of the distilled state is the normalized square of the state before distillation.
Actually, one can notice that the distillation protocol we employ for noisy GHZ state achieves similar function by comparing $\rho_{\mathrm{noise}}$ and $\rho_{\mathrm{noise}}^{'}$.

\begin{figure}[htbp]\label{fig:W_state}
\centering
\includegraphics[width=0.3\linewidth]{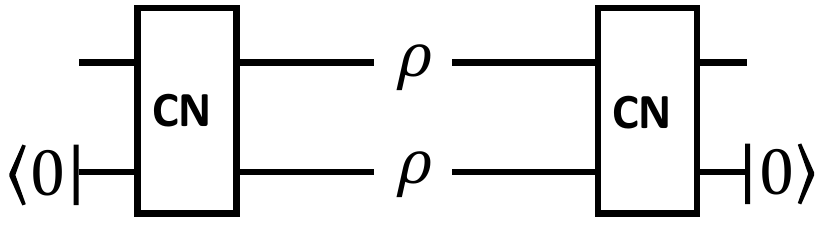}
\caption{The entanglement distillation protocol based on CNOT (CN) gates.}
\label{fig:W_state}
\end{figure}

Given the matrix form of the noisy W state
\begin{equation}
\rho_\mathrm{W}=
\begin{bmatrix}
\frac{1-p}{8} & 0 & 0 & 0 & 0 & 0 & 0 & 0\\
0 & \frac{3+5p}{24} & \frac{p}{3} & 0 & \frac{p}{3} & 0 & 0 & 0\\
0 & \frac{p}{3} & \frac{3+5p}{24} & 0 & \frac{p}{3} & 0& 0 & 0\\
0 & 0 & 0 & \frac{1-p}{8} & 0 & 0 & 0 & 0\\
0 & \frac{p}{3} & \frac{p}{3} & 0 & \frac{3+5p}{24} & 0 & 0 & 0\\
0 & 0 & 0 & 0 & 0 & \frac{1-p}{8} & 0 & 0\\
0 & 0 & 0 & 0 & 0 & 0 & \frac{1-p}{8} & 0\\
0 & 0 & 0 & 0 & 0 & 0 & 0 & \frac{1-p}{8}
\end{bmatrix},
\end{equation}
the distilled state is 
\begin{equation}
\rho_\mathrm{W}^{'}=
\frac{24}{5p^2+3}
\begin{bmatrix}
\frac{(1-p)^2}{64} & 0 & 0 & 0 & 0 & 0 & 0 & 0\\
0 & \frac{(3+5p)^2}{24^2} & \frac{p^2}{9} & 0 & \frac{p^2}{9} & 0 & 0 & 0\\
0 & \frac{p^2}{9} & \frac{(3+5p)^2}{24^2} & 0 & \frac{p^2}{9} & 0& 0 & 0\\
0 & 0 & 0 & \frac{(1-p)^2}{64} & 0 & 0 & 0 & 0\\
0 & \frac{p^2}{9} & \frac{p^2}{9} & 0 & \frac{(3+5p)^2}{24^2} & 0 & 0 & 0\\
0 & 0 & 0 & 0 & 0 & \frac{(1-p)^2}{64} & 0 & 0\\
0 & 0 & 0 & 0 & 0 & 0 & \frac{(1-p)^2}{64} & 0\\
0 & 0 & 0 & 0 & 0 & 0 & 0 & \frac{(1-p)^2}{64}
\end{bmatrix}.
\end{equation}
When the value of $p$ is large, this entanglement distillation protocol will increase the fidelity between this state and the W state.
If using the easiest way, the entanglement witness of $W_{\mathrm{W}}=\frac{2}{3}\mathbb{I}_3-\ketbra{\mathrm{W}}{\mathrm{W}}$, one can observe ``GME superactivation''.
For example, setting $p=0.6$, the expectation value of noisy W state is positive while the value for distilled state is negative.
However, for noisy W state, the entanglement witness is not the necessary and sufficient condition for GME.
Note that these both noisy W state and distilled state are all permutation-invariant three-qubit states, whose GME can be determined by the PPT-mixture criterion shown in Eq.~\eqref{eq:PPT_optimization}~\cite{jungnitsch2011taming,novo2013permutation}.
Through numerical test, we show that $\rho_{\mathrm{W}}$ is bi-separable when the value of $p$ is decreased to around $0.479$. This is actually the optimal threshold for the bi-separability of the noisy W state, which has been rigorously proven in Ref.~\cite{jungnitsch2011taming}.
At the same time, the distilled state is bi-separable when the value of $p$ is decreased to around $0.519$, showing the absence of GME superactivation.

To analyze the SLE, note that the optimal way to localize entanglement for noisy W state is to measure one qubit in computational basis and post-select the measurement result of $\ket{0}$.
Then, the post-selected state becomes
\begin{equation}
p\bra{0}\ketbra{\mathrm{W}}{\mathrm{W}}\ket{0}+\frac{1-p}{8}\mathbb{I}_2=\frac{p}{3}\left(\ket{01}+\ket{10}\right)\left(\bra{01+\bra{10}}\right)+\frac{1-p}{8}\mathbb{I}_2.
\end{equation}
After normalization and local transformation, the state becomes
\begin{equation}
\frac{4p}{p+3}\ketbra{\mathrm{EPR}}{\mathrm{EPR}}+\frac{3(1-p)}{4(p+3)}\mathbb{I}_2.
\end{equation}
When $p\ge\frac{3}{11}$, the noisy W state has SLE as the fidelity between the above state and the EPR state exists 0.5.
The distilled state can be rewritten as
\begin{equation}
\rho_\mathrm{W}^{'}=\frac{8p^2}{5p^2+3}\ketbra{\mathrm{W}}{\mathrm{W}}+\frac{3(1-p)^2}{8(5p^2+3)}\mathbb{I}_3+\frac{2p(1-p)}{5p^2+3}\left(\ketbra{001}{001}+\ketbra{010}{010}+\ketbra{100}{100}\right).
\end{equation}
According to the symmetry, measuring computational basis and post-select the result of $\ket{0}$ is also the optimal way to localize entanglement.
The resultant unnormalized state is 
\begin{equation}
\bra{0}\rho_{\mathrm{W}}^{'}\ket{0}=\frac{8p^2}{3(5p^2+3)}\left(\ket{01}+\ket{10}\right)\left(\bra{01}+\bra{10}\right)+\frac{3(1-p)^2}{8(5p^2+3)}\mathbb{I}_2+\frac{2p(1-p)}{5p^2+3}\left(\ketbra{01}{01}+\ketbra{10}{10}\right).
\end{equation}
Applying a Pauli-$X$ gate on the second qubit, this state will be transformed into 
\begin{equation}
(\mathbb{I}\otimes X)\bra{0}\rho_{\mathrm{W}}^{'}\ket{0}(\mathbb{I}\otimes X)=\frac{8p^2}{3(5p^2+3)}\left(\ket{00}+\ket{11}\right)\left(\bra{00}+\bra{11}\right)+\frac{3(1-p)^2}{8(5p^2+3)}\mathbb{I}_2+\frac{2p(1-p)}{5p^2+3}\left(\ketbra{00}{00}+\ketbra{11}{11}\right).
\end{equation}
Now we use the PPT criterion to test the entanglement of this state.
The partial transposed matrix becomes
\begin{equation}
\left[(\mathbb{I}\otimes X)\bra{0}\rho_{\mathrm{W}}^{'}\ket{0}(\mathbb{I}\otimes X)\right]^{\mathrm{T}_B}=
\begin{bmatrix}
\frac{8p^2}{3(5p^2+3)}+\frac{3(1-p)^2}{8(5p^2+3)}+\frac{2p(1-p)}{5p^2+3} & 0 & 0 & 0\\
0 & \frac{3(1-p)^2}{8(5p^2+3)} & \frac{8p^2}{3(5p^2+3)} & 0\\
0 & \frac{8p^2}{3(5p^2+3)} & \frac{3(1-p)^2}{8(5p^2+3)} & 0\\
0 & 0 & 0 & \frac{8p^2}{3(5p^2+3)}+\frac{3(1-p)^2}{8(5p^2+3)}+\frac{2p(1-p)}{5p^2+3}
\end{bmatrix}.
\end{equation}
Note that when the value of $\frac{8p^2}{3(5p^2+3)}$ is larger than $\frac{3(1-p)^2}{8(5p^2+3)}$, the partial transposed matrix is not semi-positive.
One can verify that similarly when $p>\frac{3}{11}$, the above matrix is not semi-positive.
Thus, the noisy W state does not possess SLE superactivation phenomenon.

Through this analysis, we found that compared with noisy GHZ state, the noisy W state is hard to demonstrate the entanglement superactivation for both GME and SLE.
This result shows an intriguing relationship between entanglement superactivation and state inter-convertibility under SLOCC operations.
It is natural to ask whether a better distillation protocol exists that can demonstrate the entanglement superactivation with noisy W state (for a detailed discussion see Ref.~\cite{weinbrenner2024superactivation}); or this is the inherent property of W state.
Besides, as multipartite entanglement has a much more complicated structure, it is worthy to explore the entanglement superactivation in larger systems.

\section{Difference between SLE and GME}
In Sec.~\ref{subsec:noisy_GHZ}, we use the noisy GHZ state to demonstrate the existence of states possessing SLE but not GME. 
This naturally raises the question of whether there also exist states with only GME, but not SLE.
Ref.~\cite{gunhe2016emergent} provides an affirmative answer by numerically identifying a three-qubit genuinely multipartite entangled state whose two-qubit reduced state, after any projective measurement on the remaining qubit, is always separable. 
Combined with Theorem~\ref{theorem:main}, this implies that the identified three-qubit state exhibits GME without SLE. Therefore, SLE and GME should be regarded as two distinct notions of multipartite entanglement, without any inclusion relation between them.

\subsection{SLE superactivation and the PPT square conjecture}
After clarifying the relationship between SLE and GME, it is natural and important to further examine the connection between SLE superactivation and GME superactivation. It is known that for a multipartite quantum state $\rho$ that is not separable in any fixed bipartition, its GME can always be activated given sufficiently many copies~\cite{Palazuelos2022genuinemultipartite}. Equivalently, for GME, a multipartite state can be activated if and only if it is not separable in some fixed bipartition. This raises the question of whether the same condition is also necessary and sufficient for SLE superactivation. If not, can we identify multipartite quantum states that are entangled across every bipartition, but yet cannot be activated to be SLE?

The PPT square conjecture, if it is assumed to be correct, provides such a counterexample. 
The conjecture was, to our knowledge
first formulated by M. Christandl and states that the concatenation of two PPT channels is entanglement-breaking~\cite{ruskai2012operator,chen2019pptsquare}. Using the Choi isomorphism, this is equivalent to saying that for two arbitrary bipartite PPT quantum states $\rho_{AB}$ and $\rho_{CD}$, the state
\begin{equation}
\sigma_{AD}=\bra{\Phi^+_{BC}}\left(\rho_{AB}\otimes\rho_{CD}\right)\ket{\Phi^+_{BC}},
\end{equation}
where $\ket{\Phi_{BC}^+}$ is the maximally entangled state defined on parties $B$ and $C$, is a separable state on parties $A$ and $D$.
Note that $\rho_{AB}$ and $\rho_{CD}$ may themselves be PPT entangled states. 
Finally, one may also rephrase the PPT square conjecture as the statement that PPT entangled states are useless for entanglement swapping.

When we change the maximally entangled state $\ket{\Phi^+_{BC}}$ to some other pure state $\ket{\Psi_{BC}}$, the resultant state is also proportional to a separable state.
This can be proved from the simple observation that any bipartite pure state can be written in the form $\ket{\Psi_{BC}}=\mathbb{I}_B\otimes K_C\ket{\Phi_{BC}^+}$, where $K_C$ is a general matrix acting on system $C$.
Then, we have
\begin{equation}\label{eq:ppt_psi}
\bra{\Psi_{BC}}\left(\rho_{AB}\otimes\rho_{CD}\right)\ket{\Psi_{BC}}=\bra{\Phi^+_{BC}}\left[\rho_{AB}\otimes\left(K_C\otimes \mathbb{I}_D\right)\rho_{CD}\left(K_C^\dagger\otimes \mathbb{I}_D\right)\right]\ket{\Phi^+_{BC}}.
\end{equation}
It is easy to check that, when $\rho_{CD}$ is a PPT state, then $\left(K_C\otimes \mathbb{I}_D\right)\rho_{CD}\left(K_C^\dagger\otimes \mathbb{I}_D\right)$ is also proportional to a PPT state.
Therefore, assuming the correctness of the PPT square conjecture, $\bra{\Psi_{BC}}\left(\rho_{AB}\otimes\rho_{CD}\right)\ket{\Psi_{BC}}$ is proportional to a separable state for arbitrary $\ket{\Psi_{BC}}$.

Based on this observation, we can construct a state that is not separable in any fixed bi-partition while cannot be activated in SLE.
Consider a three-partite state $\rho_{ABC}=\sigma_{AB_1}\otimes \sigma_{B_2C}$, where $B_1$ and $B_2$ are two parties of system $B$ and $\sigma_{AB_1}$ and $\sigma_{B_2C}$ are PPT  entangled states.
It can be easily verified that $\rho_{ABC}$ is entangled in any bi-partition as $\sigma_{AB_1}$ and $\sigma_{B_2C}$ are entangled states.
When preparing many copies of it, the joint state can be written in the form of $\rho_{ABC}^{\otimes t}=\sigma_{AB_1}^{\otimes t}\otimes \sigma_{B_2C}^{\otimes t}$.
As the partial transposition operation is tensor stable, the joint states $\sigma_{AB_1}^{\otimes t}$ and $\sigma_{B_2C}^{\otimes t}$ are also PPT in the partition of $A|B_1$ and $B_2|C$.
Combining the PPT square conjecture and the observation made in Eq.~\eqref{eq:ppt_psi}, we thus have
\begin{equation}
    \sigma_{AC}=\frac{\bra{\Psi_{B}}\rho_{ABC}^{\otimes t}\ket{\Psi_{B}}}{\mathrm{Tr}\bra{\Psi_{B}}\rho_{ABC}^{\otimes t}\ket{\Psi_{B}}}
\end{equation}
is a separable state for arbitrary $t$ and $\ket{\Psi_{B}}$.
According to Theorem~\ref{theorem:main}, this means that $\rho_{ABC}^{\otimes t}$ does not have SLE for arbitrary value of $t$.
Therefore, we have proven that, if the PPT square conjecture is correct, there exist some multipartite state that is not separable in any bipartition while cannot be SLE activated.

A corollary of this result is that SLE superactivation on all bipartite subsystems is strictly more difficult than GME superactivation. 
This is because a state can only exhibit SLE superactivation on all bipartitions if it is not separable across any fixed bipartition, a necessary and sufficient condition for GME superactivation. 
Here, SLE superactivation on all bipartitions means there exists an integer $t$ such that $\rho^{\otimes t}$ has SLE on every bipartite subsystem.

Another direct corollary is that there exist genuinely multipartite entangled state that cannot be activated in SLE.
We can still focus on the state mentioned before, $\rho_{ABC}$, which is not separable under any fixed bipartition.
According to the result of Ref.~\cite{Palazuelos2022genuinemultipartite}, there exists a value of $t_1$, such that $\rho_{ABC}^{\otimes t_1}$ is genuinely multipartite entangled.
While, as $\rho_{ABC}^{\otimes t_1}$ can also be written as the tensor product of two PPT entangled states, it cannot be SLE activated. 

\section{Experimental details}
\subsection{Analysis of experimental distillation for noisy GHZ states}\label{sec:distillation_protocol}
Previous theoretical and experimental works have demonstrated that bipartite entanglement distillation in polarization degrees of freedom of photons can be achieved with polarizing beam splitters (PBSs) \cite{Pan2001purification,Pan2003purification,chen2017experimental}.
Here, we present an analysis of the experimental tripartite distillation process for the three-photon noisy GHZ state prepared in polarization degrees of freedom.

The to-be-distilled noisy GHZ state can be decomposed into a mixture of eight GHZ states in Eq.~\eqref{eq:eight_GHZ},
where states $\ket{0}$ and $\ket{1}$ are encoded in $\ket{H}$ and $\ket{V}$ polarization of photons, respectively.
Proportions of all these component states are $F_{0}^{+}=\frac{1+7p}{8}$ and $F_{0}^{-}=F_{1}^{+}=...=F_{3}^{-}=\frac{1-p}{8}=F_{r}$, which are also fidelities between the noisy GHZ state and these component states. 
Among them, the state $\ket{G_{0}^{+}}$, with the largest fidelity of $F_{0}^{+}$, is the target state for the entanglement distillation.

As illustrated in Figure~2 in the main text, to perform the distillation operation, three clients superimpose their two photons on a PBS and post-select events where there is exactly one photon in each output of PBS. 
This requires the two photons to be in the same polarization, i.e., both $\ket{0}$ or both $\ket{1}$. 
As a result, the PBS operation, together with the post-selection, can be described using the projection operator $P=\ketbra{00}{00}+\ketbra{11}{11}$.
The tripartite distillation allows two component states of two to-be-distilled noisy GHZ states to meet randomly on PBSs. 
From the polarization distribution of eight component states in Eq.~\eqref{eq:eight_GHZ}, we can see that the post-selection condition is satisfied only when the subscripts of two meeting component states are consistent, $\ket{G_{i}^{+}}$ and $\ket{G_{i}^{+}}$, $\ket{G_{i}^{+}}$ and $\ket{G_{i}^{-}}$, or $\ket{G_{i}^{-}}$ and $\ket{G_{i}^{-}}$.
Other possible combinations of component states with different subscripts will be eliminated by post-selection, such as $\ket{G_0^+}$ meets $\ket{G_1^+}$.
This can be verified mathematically by $P^{\otimes 3}\ket{G_0^+}\ket{G_1^+}=0$.

After the post-selection, each client measures one of two output photons on the Pauli-$X$ basis and compares the result. 
The Pauli-$X$ measurement transforms the projection operator into either $P_+=\bra{+}P=\frac{1}{\sqrt{2}}(\ketbra{0}{00}+\ketbra{1}{11})$ or $P_-=\bra{-}P=\frac{1}{\sqrt{2}}(\ketbra{0}{00}-\ketbra{1}{11})$.
As stated in the main text, three clients keep the resultant three-photon state when an even number of $\ket{-}$ is recorded and apply a phase-flip operation when an odd number of $\ket{-}$ is recorded.
Therefore, if $\ket{G_i^+}$ meets $\ket{G_i^+}$ or $\ket{G_i^-}$ meets $\ket{G_i^-}$ on PBSs, the resultant state will be $\ket{G_i^+}$; 
while if $\ket{G_i^+}$ meets $\ket{G_i^-}$, the resultant state will be $\ket{G_i^-}$.
This can be verified with 
\begin{equation}
\begin{aligned}
&(P_+P_+P_+ + P_+P_-P_- + P_-P_+P_- + P_-P_-P_+)\ket{G_i^+}\otimes\ket{G_i^+}\propto\ket{G_i^+},\\
&(P_+P_+P_+ + P_+ P_- P_- + P_- P_+ P_- + P_- P_- P_+)\ket{G_i^+}\otimes\ket{G_i^-}\propto\ket{G_i^-},\\
&(P_- P_- P_- + P_- P_+ P_+ + P_+ P_- P_+ + P_+ P_+ P_-)\ket{G_i^+}\otimes\ket{G_i^+}\propto\ket{G_i^-},\\
&(P_- P_- P_- + P_- P_+ P_+ + P_+ P_- P_+ + P_+ P_+ P_-)\ket{G_i^+}\otimes\ket{G_i^-}\propto\ket{G_i^+}.
\end{aligned}
\end{equation}

Thus, after this whole distillation process, one obtains a new mixed state, which is a mixture of the same component states with new proportions $F_{0}^{+\prime}$, $F_{0}^{-\prime}$, $\cdots$, $F_{3}^{-\prime}$.
For example, the fidelity for the target component state $\ket{G_{0}^{+}}$ becomes
\begin{equation}
F_{0}^{+\prime}=\frac{(F_{0}^{+})^2 +(F_{0}^{-})^2 }{ (F_{0}^{+}+F_{0}^{-})^2+3\times (2 F_{r})^2}
=\frac{(\frac{1+7p}{8})^2 +(\frac{1-p}{8})^2 }{ (\frac{1+7p}{8}+\frac{1-p}{8})^2+12\times (\frac{1-p}{8})^2}
=\frac{25p^2+6p+1}{24p^2+8},
\label{eq:F'}
\end{equation}
where the nominator is the probability of cases that $\ket{G_0^+}$ meets $\ket{G_0^+}$ and $\ket{G_0^-}$ meets $\ket{G_0^-}$, and the denominator represents the probability of the success of post-selection.
One can prove that, for $p\in (\frac{4\sqrt{3}-3}{13}, \frac{3}{7}) \sim(0.3022,0.4286)$, it holds that ${F_{0}^{+}}<0.5$ while $F_{0}^{+\prime}>0.5$. 
This means the distilled state is genuinely entangled, while the to-be-distilled states are bi-separable. 
In other words, one can obtain genuine tripartite entanglement from two non-genuinely-entangled noisy GHZ states. This conclusion and the corresponding value interval are consistent with the theoretical analysis in Sec.~\ref{sec:ideal_Werner}.

\begin{figure}[htbp!]
	\centering
	\includegraphics[width=0.8 \linewidth]{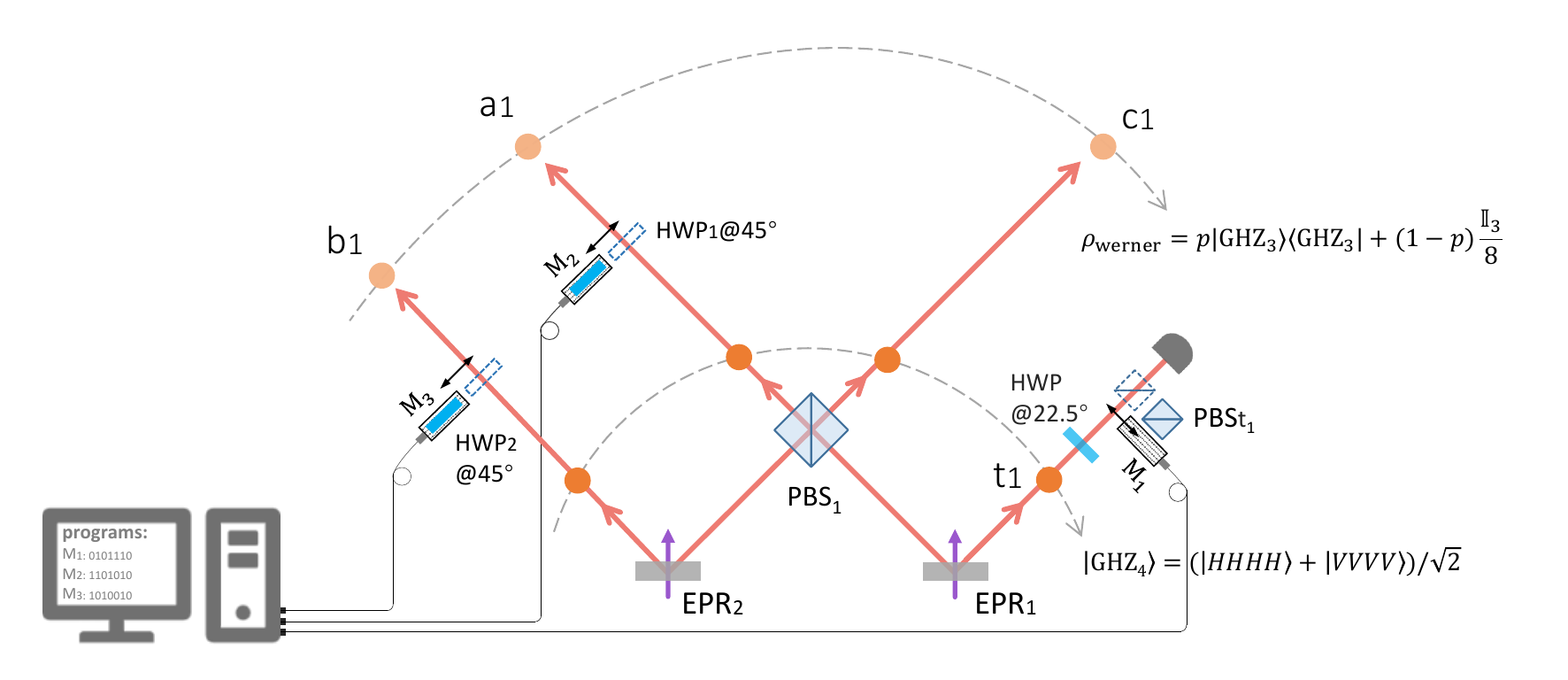}
	\caption{Schematic diagram for the preparation of the to-be-distilled noisy GHZ state.}
	\label{Fig:PreparationWerner}
\end{figure}

\subsection{Generation procedure of noisy GHZ states} \label{sec:preparation}

Here, we describe the preparation of the noisy GHZ state in our experiment, using $\rho_{1}$ as an example. 
The procedure is illustrated in Fig.~\ref{Fig:PreparationWerner}. 
First, we generate a four-photon GHZ state $\ket{\text{GHZ}_4}=\frac{1}{\sqrt{2}}(\ket{HHHH}+\ket{VVVV})$ by interfering photons from EPR$_1$ and EPR$_2$ on PBS$_1$.
Based on this state, we perform different measurements on the heralding photon t$_1$ to prepare different component states of the noisy GHZ state on photons a$_1$-b$_1$-c$_1$. 
Specifically, by introducing a HWP@22.5$\degree$ in the light path of t$_1$ in advance, the following two scenarios arise:

(1) When PBS$_{\text{t}_1}$ is in the light path, and HWP$_1$@$45\degree$ and HWP$_2$@$45\degree$ are out of light paths, the heralding photon t$_1$ will be triggered to state $\ket{+}$, and three photons a$_1$-b$_1$-c$_1$ will be prepared as $\ket{G_{0}^{+}}=\frac{1}{\sqrt{2}}(\ket{HHH}+\ket{VVV})$. 
This is the pure state part of the noisy GHZ state (See row 1 of Table~\ref{tab:preparation}).

(2) When PBS$_{\text{t}_1}$ is removed from the light path, the heralding photon t$_1$ directly enters the detector without polarization identification and is triggered to a mixed state $\rho_{\text{t}_1}=\frac{1}{2}(\ketbra{+}+\ketbra{-})$. 
In this case, photons a$_1$-b$_1$-c$_1$ are prepared as the mixed state $\rho_{G_{0}}=\frac{1}{2}\ketbra{HHH}+\frac{1}{2}\ketbra{VVV}=\frac{1}{2}\ketbra{G_{0}^{+}}+\frac{1}{2}\ketbra{G_{0}^{-}}$. 
On this basis, we randomly insert two HWP@45$\degree$ into light paths of photons a$_1$ and b$_1$ with the probability of 0.5, yielding three other mixed states with bit flip, $\rho_{G_{1}}=\frac{1}{2}\ketbra{G_{1}^{+}}+\frac{1}{2}\ketbra{G_{1}^{-}}$, $\rho_{G_{2}}=\frac{1}{2}\ketbra{G_{2}^{+}}+\frac{1}{2}\ketbra{G_{2}^{-}}$, and $\rho_{G_{3}}=\frac{1}{2}\ketbra{G_{3}^{+}}+\frac{1}{2}\ketbra{G_{3}^{-}}$. 
By randomly and equally preparing these four mixed states on photons a$_1$-b$_1$-c$_1$, we effectively prepare the three-photon maximally mixed state $\frac{\mathbb{I}_3}{8}$.
This is the maximally mixed state part of the noisy GHZ state (See rows 2-5 of Table~\ref{tab:preparation}).

By combining above scenarios randomly with a specific probability relevant to the parameter $p$, the noisy GHZ state $\rho_{\mathrm{noise}}=p \ketbra{G_{0}^{+}}+(1-p)\frac{\mathbb{I}_3}{8}$ can be prepared. 
The detailed experimental settings are listed in Table~\ref{tab:preparation}. 
The movement of the HWPs and PBSs is controlled by a set of stepping motors, which are programmed by a central control computer to ensure the optimal mixture of the component states.
In this way, given an average generation rate $g$ for EPR sources, the generation probability of the noisy GHZ state is calculated as $\frac{g^2}{2(1+p)}$, primarily limited by the SPDC process.


\begin{table}[!htbp]
	\renewcommand{\arraystretch}{1.5}
	\normalsize 
	\centering
	\begin{tabular}{c|c c c|c|c|c|c}
	\hline
    \hline
    \multicolumn{4}{c|}{~\textbf{Component states of the noisy GHZ state}~} & \multicolumn{4}{c}{~~\textbf{Stepping motor group settings}~~}\\
    \cline{5-8}	
	\multicolumn{4}{c|}{~$\rho_\mathrm{noise}=p\ketbra{\mathrm{GHZ}_3}{\mathrm{GHZ}_3}+(1-p)\frac{\mathbb{I}_3}{8}$~} & {\textbf{M1:PBS$_\text{t1}$}} & {\textbf{M2:HWP$_1$@$45\degree$}} & {\textbf{M3:HWP$_2$@$45\degree$}} & {\textbf{Probability}}\\
        \hline
        $\ket{\mathrm{GHZ}_3}$ & \multicolumn{3}{c|}{$\ket{G_{0}^{+}}=\frac{1}{\sqrt{2}}(\ket{HHH}+\ket{VVV})$} & ${\textit{\color{black} in}}$ & ${\textit{\color{black} out}}$ & ${\textit{\color{black} out}}$ & \textbf{\color{black} $p'=\frac{2p}{1+p}$}\\
		\hline
        \multirow{8}{*}{\Large{$\frac{\mathbb{I}_3}{8}$}} & \multirow{1}{*}{$\rho_{G_{0}}=\frac{1}{2}\ketbra{HHH}$}  & \multirow{2}{*}{$\Rightarrow$} & \multirow{1}{*}{$\rho_{G_{0}}=\frac{1}{2}\ketbra{G_{0}^{+}}$} & \multirow{2}{*}{\textit{\color{black} out}} & \multirow{2}{*}{\textit{\color{black} out}} & \multirow{2}{*}{\textit{\color{black} out}} & \multirow{2}{*}{\textbf{\color{black} $\frac{1-p'}{4}$}}\\
         & \multirow{1}{*}{$~~~~~+\frac{1}{2}\ketbra{VVV}$} &  & \multirow{1}{*}{$~~~~~~+\frac{1}{2}\ketbra{G_{0}^{-}}$} &  &  &  & \\
         \cline{2-8}	
         & \multirow{1}{*}{$\rho_{G_{1}}=\frac{1}{2}\ketbra{HHV}$}  & \multirow{2}{*}{$\Rightarrow$} & \multirow{1}{*}{$\rho_{G_{1}}=\frac{1}{2}\ketbra{G_{1}^{+}}$} & \multirow{2}{*}{\textit{\color{black} out}} & \multirow{2}{*}{\textit{\color{black} in}} & \multirow{2}{*}{\textit{\color{black} out}} & \multirow{2}{*}{\textbf{\color{black} $\frac{1-p'}{4}$}} \\
         & \multirow{1}{*}{$~~~~~+\frac{1}{2}\ketbra{VVH}$} &  & \multirow{1}{*}{$~~~~~~+\frac{1}{2}\ketbra{G_{1}^{-}}$} &  &  &  & \\
         \cline{2-8}	
         & \multirow{1}{*}{$\rho_{G_{2}}=\frac{1}{2}\ketbra{HVH}$}  & \multirow{2}{*}{$\Rightarrow$} & \multirow{1}{*}{$\rho_{G_{2}}=\frac{1}{2}\ketbra{G_{2}^{+}}$} & \multirow{2}{*}{\textit{\color{black} out}} & \multirow{2}{*}{\textit{\color{black} out}} & \multirow{2}{*}{\textit{\color{black} in}} &  \multirow{2}{*}{\textbf{\color{black} $\frac{1-p'}{4}$}}\\
         & \multirow{1}{*}{$~~~~~+\frac{1}{2}\ketbra{VHV}$} &  & \multirow{1}{*}{$~~~~~~+\frac{1}{2}\ketbra{G_{2}^{-}}$} &  &  &  & \\
         \cline{2-8}		         
         & \multirow{1}{*}{$\rho_{G_{3}}=\frac{1}{2}\ketbra{VHH}$}  & \multirow{2}{*}{$\Rightarrow$} & \multirow{1}{*}{$\rho_{G_{3}}=\frac{1}{2}\ketbra{G_{3}^{+}}$} & \multirow{2}{*}{\textit{\color{black} out}} & \multirow{2}{*}{\textit{\color{black} in}} & \multirow{2}{*}{\textit{\color{black} in}} &  \multirow{2}{*}{\textbf{\color{black} $\frac{1-p'}{4}$}} \\
         & \multirow{1}{*}{$~~~~~+\frac{1}{2}\ketbra{HVV}$} &  & \multirow{1}{*}{$~~~~~~+\frac{1}{2}\ketbra{G_{3}^{-}}$} &  &  &  & \\
         \cline{2-8}	     
        \hline
             \hline
	\end{tabular}
    \caption{Experimental settings for preparation of the noisy GHZ state. ``\textit{in}" denotes that the optical component is inserted into the light paths, while ``\textit{out}" denotes that the optical component is removed from the light paths.}
    \label{tab:preparation}
\end{table}

\subsection{Exclusion of same-order double-pair emission}
\label{sec:double_pair}
The faithful entanglement distillation demonstration requires preparing two copies of noisy GHZ states using four EPR sources, each generating exactly one photon pair, denoted as $g_{1}$-$g_{2}$-$g_{3}$-$g_{4}$ (See rows 1-2 of Table~\ref{tab:double_pair}).
Success events are registered as eight-body coincidence counts, where exactly one photon is detected at each of the eight detection terminals simultaneously. 
However, the probabilistic nature of SPDC processes may lead to same-order double-pair emission events, where some EPR sources yield two photon pairs while some yield none.
Some of these events, such as  $g^{2}_{1}$-$0$-$g_{3}$-$g_{4}$ and $g_{1}$-$g^{2}_{2}$-$0$-$g_{4}$, have the same probability as the ideal case and may be registered, thus deviating the experiment from a faithful distillation procedure \cite{Pan2003purification,chen2017experimental}.

In our experiment, as depicted in Fig.~\ref{Fig:doublepair}, we meticulously design the entanglement distillation network in a crossed structure. 
This design ensures that all same-order double-pair emission events are automatically filtered out, with only ideal events producing eight-body coincidence counts. The details are as follows.

(1) Photons t$_1$ and t$_2$, from EPR$_1$ and EPR$_4$, serve as heralding photons, ensuring that each of these two EPR sources generates exactly one photon pair.

(2) The arrangement of distillation PBSs excludes other double-pair emission events of EPR$_2$ or EPR$_3$, such as $g_{1}$-$g^{2}_{2}$-$0$-$g_{4}$ or $g_{1}$-$0$-$g^{2}_{3}$-$g_{4}$. 
Due to the symmetry of the experimental setup, we only discuss the former case here, with some typical examples listed in Table~\ref{tab:double_pair}. 
(i) When two photon pairs are produced from double-pair emission with the same polarization (e.g., $\ket{H_{b_1}H_{b_1}H_{c_1}H_{c_1}}$ or $\ket{V_{b_1}V_{b_1}V_{a_1}V_{a_1}}$), the two photons will enter the same detection terminal and cannot produce an eight-body coincidence count, such as events in rows 3-4 of Table~\ref{tab:double_pair}.
(ii) When two photon pairs are produced from double-pair emission with different polarizations (e.g., $\ket{H_{b_1}V_{b_1}H_{c_1}V_{a_1}}$), there is always one detection terminal that receives no photons, preventing an eight-body coincidence count, such as events in rows 5-8 of Table~\ref{tab:double_pair}.

The above discussion is based on the assumption that no HWPs are inserted into the light paths. 
In our experiment, we prepare the noisy GHZ state by inserting and removing HWP@$45\degree$ in the light paths of photons $a$ and $b$. 
During the generation procedure of noisy GHZ states, the filtering strategy described above remains effective.

\begin{figure}[htbp!]
	\centering
	\includegraphics[width=0.5 \linewidth]{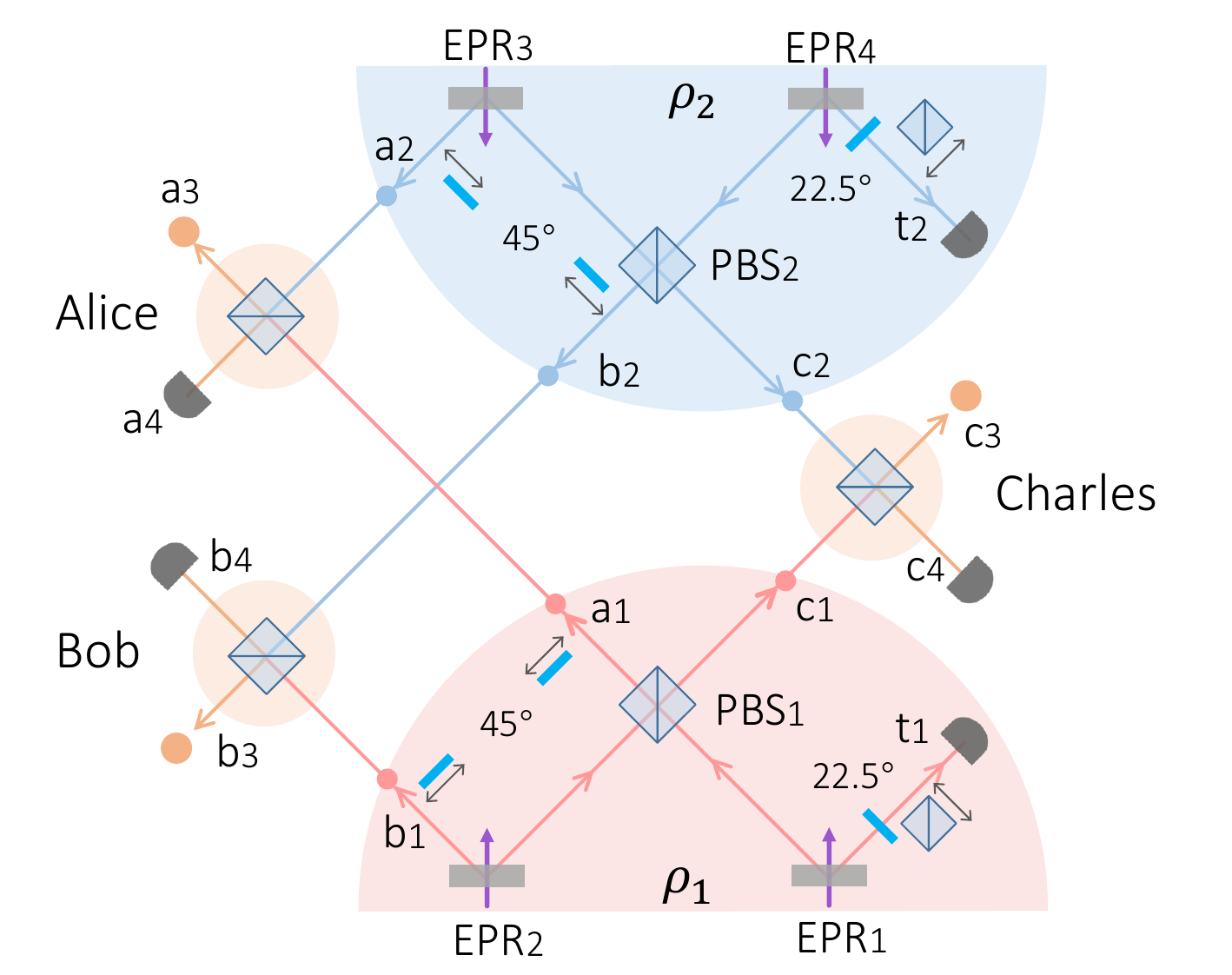}
	\caption{Schematic diagram of the entanglement distillation network.}
	\label{Fig:doublepair}
\end{figure}

\begin{table}[!htbp]
	\renewcommand{\arraystretch}{1.5}
	\normalsize 
	\centering	\begin{tabular}{c|c|c|c|c|c|c|c|c|c|c|c|c|c}
\hline
\hline
\multicolumn{5}{c|}{\textbf{SPDC process}} & \multicolumn{8}{c|}{\textbf{photon distribution}} & \textbf{eight-body}  \\ \cline{1-13} 
\textbf{same-order emission} & \textbf{EPR$_1$} & \textbf{EPR$_2$} & \textbf{EPR$_3$} & \textbf{EPR$_4$} & t$_1$ & t$_2$ & a$_3$ & a$_4$ & b$_3$ & b$_4$ & c$_3$ & c$_4$ &  \textbf{coincidence?} \\ \hline
\multirow{2}{*}{$g_{1}$-$g_{2}$-$g_{3}$-$g_{4}$}& $\ket{H_{t_1}H_{a_1}}$ & $\ket{H_{b_1}H_{c_1}}$ & $\ket{H_{a_2}H_{c_2}}$ & $\ket{H_{t_2}H_{b_2}}$ & 1 & 1 & 1 & 1 & 1 & 1 & 1 & 1 & \checkmark \\ 
\cline{2-14}
 & $\ket{V_{t_1}V_{c_1}}$ & $\ket{V_{b_1}V_{a_1}}$ & $\ket{V_{a_2}V_{b_2}}$ & $\ket{V_{t_2}V_{c_2}}$ & 1 & 1 & 1 & 1 & 1 & 1 & 1 & 1 & \checkmark \\ \hline
\multirow{6}{*}{$g_{1}$-$g^{2}_{2}$-$0$-$g_{4}$} & $\ket{H_{t_1}H_{a_1}}$ & $\ket{H_{b_1}H_{b_1}H_{c_1}H_{c_1}}$ & 0 & $\ket{H_{t_2}H_{b_2}}$ & 1 & 1 & 1 & 0 & 1 & 2 & 2 & 0 & $\times$ \\ \cline{2-14}
 & $\ket{H_{t_1}H_{a_1}}$ & $\ket{H_{b_1}H_{b_1}H_{c_1}H_{c_1}}$ & 0 & $\ket{V_{t_2}V_{c_2}}$ & 1 & 1 & 1 & 0 & 0 & 2 & 3 & 0 & $\times$ \\ \cline{2-14}
 & $\ket{H_{t_1}H_{a_1}}$ & $\ket{H_{b_1}V_{b_1}H_{c_1}V_{a_1}}$ & 0 & $\ket{H_{t_2}H_{b_2}}$ & 1 & 1 & 1 & 1 & 2 & 1 & 1 & 0 & $\times$ \\ \cline{2-14}
 & $\ket{V_{t_1}V_{c_1}}$ & $\ket{H_{b_1}V_{b_1}H_{c_1}V_{a_1}}$ & 0 & $\ket{V_{t_2}V_{c_2}}$ & 1 & 1 & 0 & 1 & 1 & 1 & 2 & 1 & $\times$ \\ \cline{2-14}
 & $\ket{H_{t_1}H_{a_1}}$ & $\ket{H_{b_1}V_{b_1}H_{c_1}V_{a_1}}$ & 0 & $\ket{V_{t_2}V_{c_2}}$ & 1 & 1 & 1 & 1 & 1 & 1 & 2 & 0 & $\times$ \\ \cline{2-14}
 & $\ket{V_{t_1}V_{c_1}}$ & $\ket{H_{b_1}V_{b_1}H_{c_1}V_{a_1}}$ & 0 & $\ket{H_{t_2}H_{b_2}}$ & 1 & 1 & 0 & 1 & 2 & 1 & 1 & 1 & $\times$ \\ 
\hline
\hline
\end{tabular}
    \caption{Photon number distribution at eight detection terminals for the same-order SPDC process. 
    Here are just some typical same-order double-pair emission events. 
    Events for  $g_{1}$-$0$-$g^{2}_{3}$-$g_{4}$ are omitted due to the symmetry of the experimental setup.   
    }
    \label{tab:double_pair}
\end{table}

\subsection{Fidelity estimation}
As described in the main text, the three-qubit state we investigate is a mixture of a GHZ state and white noise. Therefore, for a given three-qubit state to be measured, its fidelity is defined with respect to the ideal GHZ state
\begin{equation}
	F =\bra{\text{GHZ}_3} \rho \ket{\text{GHZ}_3}= \text{Tr}(\rho\cdot \hat{\rho}_{\text{GHZ}}),
\end{equation}
Here, the density operator $\hat{\rho}_{\text{GHZ}}$ can be decomposed as follows:
\begin{align*}
\hat{\rho}_{\text{GHZ}} &= \ket{\text{GHZ}_3}\bra{\text{GHZ}_3} \\
&= \frac{1}{2}(\ket{HHH}\bra{HHH} + \ket{VVV}\bra{VVV}) + \frac{1}{2}(\ket{HHH}\bra{VVV} + \ket{VVV}\bra{HHH}) \\
&= \frac{1}{2}(\ket{H}\bra{H}^{\otimes 3} + \ket{V}\bra{V}^{\otimes 3}) + \frac{1}{6}\sum_{k=0}^{2}(-1)^k M_k^{\otimes 3}
\end{align*}
where the operator $M_k = \cos\left(\frac{k\pi}{3}\right)\sigma_x + \sin\left(\frac{k\pi}{3}\right)\sigma_y$. 
Thus, to evaluate the fidelity of each three-qubit state, we experimentally measure the observables $\sigma_z\sigma_z\sigma_z$, $M_0$, $M_1$, and $M_2$.

For the two-photon states after the localization operation, their fidelities are defined with respect to the Bell state $\ket{\Phi^+} = (\ket{HH} + \ket{VV})/\sqrt{2}$, and the fidelity can be estimated as
\begin{equation}
	F =\bra{\text{Bell}} \rho \ket{\text{EPR}}= \text{Tr}(\rho\cdot \hat{\rho}_{\text{EPR}}),
\end{equation}
the density operator $\hat{\rho}_{\text{EPR}}$ can be decomposed as follows:
\begin{equation}
\hat{\rho}_{\text{EPR}} = \frac{1}{2}(\ket{H}\bra{H}^{\otimes 2} + \ket{V}\bra{V}^{\otimes 2}) + \frac{1}{4}(\sigma_x^{\otimes 2} - \sigma_y^{\otimes 2})
\end{equation}
Therefore, to evaluate the fidelity of each two-qubit state, we experimentally measure the observables $\sigma_z\sigma_z$, $\sigma_x\sigma_x$, and $\sigma_y\sigma_y$.

\subsection{Tomographic Data}\label{sec:tomo}

This section presents the reconstructed density matrices of to-be-distilled states with $p=1$, $p=0.5$, and $p=0.36$. 
Fig.~\ref{Fig:tomo_1.0} can benchmark the state preparation process, like estimate parameters of $q$ and $r$ defined in Sec.~\ref{sec:SLE_GME_noisy}. Fig.~\ref{Fig:tomo_0.5} and Fig.~\ref{Fig:tomo_0.36} are employed to certify the GME and SLE superactivation phenomena, as discussed in Sec.~\ref{sec:SLE_GME_noisy}, respectively.
For each to-be-distilled state, $27$ measurement settings ($\sigma_x\sigma_x\sigma_x$, $\sigma_x\sigma_x\sigma_y$, $\sigma_x\sigma_x\sigma_z$, $\cdots$, $\sigma_z\sigma_z\sigma_z$) are applied.
Data are collected for 5 minutes for each measurement setting for Fig.~\ref{Fig:tomo_1.0}, and 3 hours for each measurement setting for Fig.~\ref{Fig:tomo_0.5} and Fig.~\ref{Fig:tomo_0.36}.

\begin{figure*}[htbp]
	\centering	\includegraphics[width=1 \linewidth]{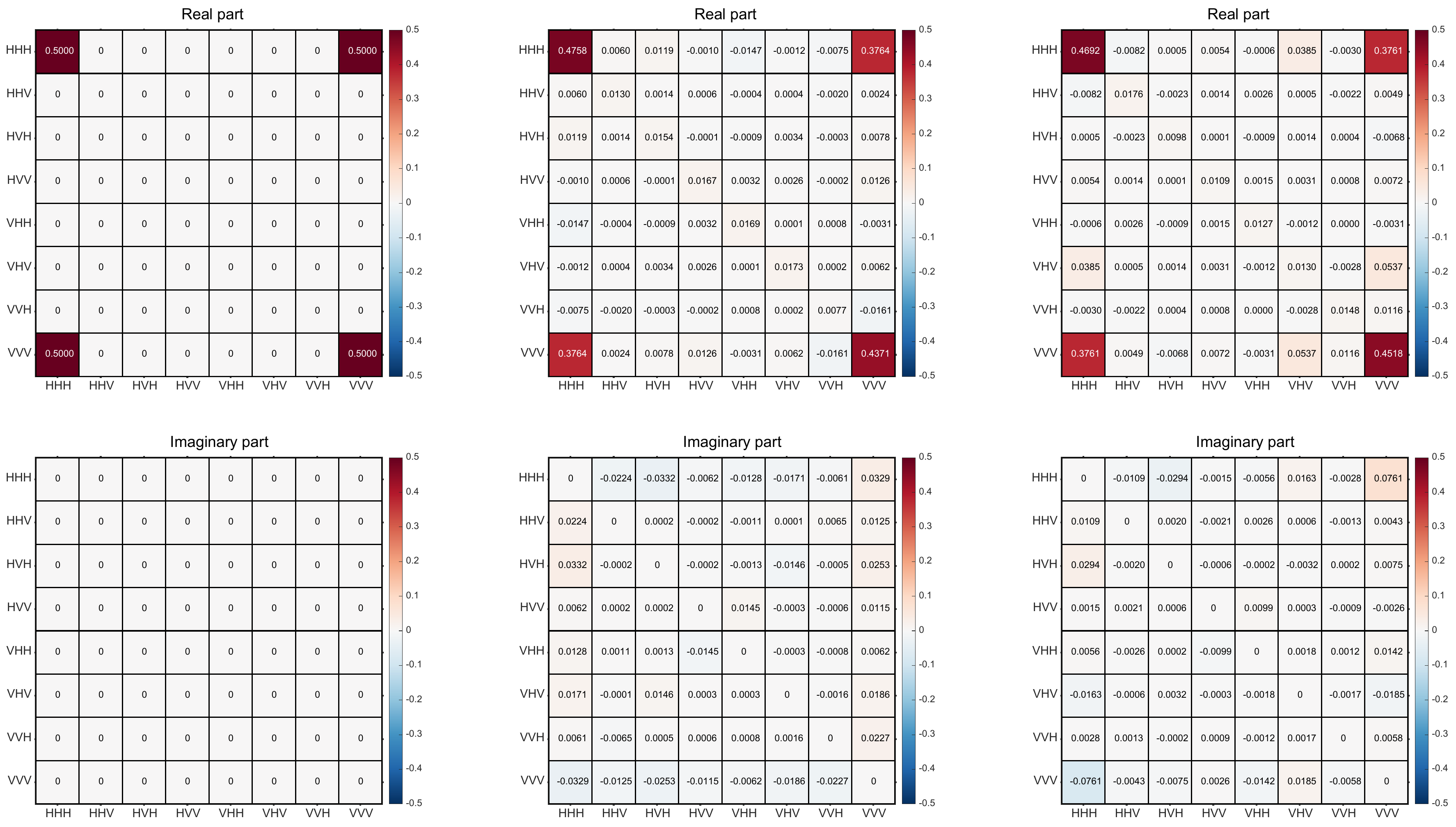}
 \caption{
\textbf{Left:} The density matrix of the noisy GHZ state with $p=1$. 
\textbf{Middle:} Experimental reconstructed density matrix of $\rho_{1}$ with $p=1$. 
\textbf{Right:} Experimental reconstructed density matrix of $\rho_{2}$ with $p=1$.
}
\label{Fig:tomo_1.0}
\end{figure*}

\begin{figure*}[htbp]
	\centering
	\includegraphics[width=1 \linewidth]{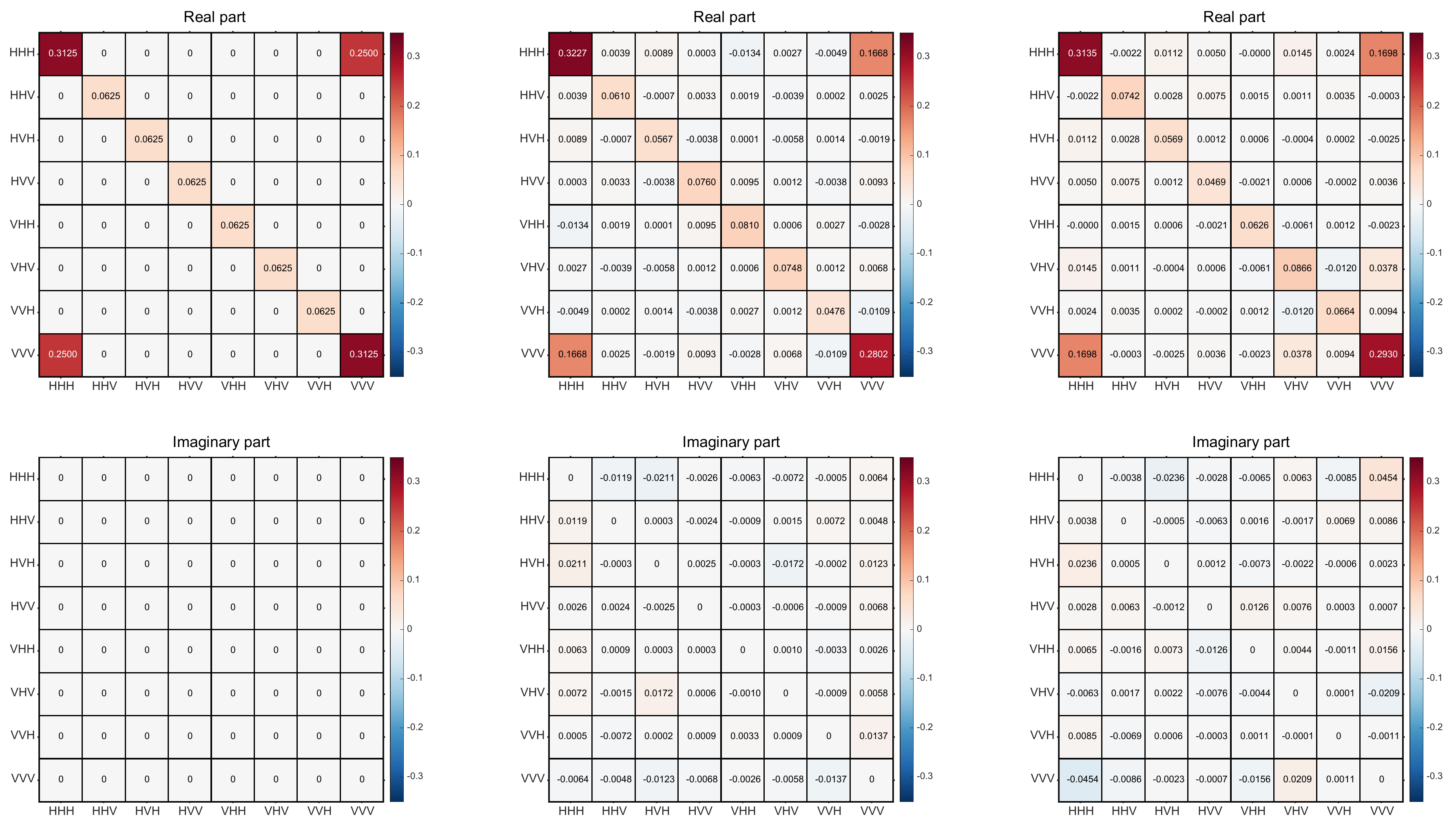}
	\caption{
\textbf{Left:} The density matrix of the noisy GHZ state with $p=0.5$. 
\textbf{Middle:} Experimental reconstructed density matrix of $\rho_{1}$ with $p=0.5$. 
\textbf{Right:} Experimental reconstructed density matrix of $\rho_{2}$ with $p=0.5$.
}
	\label{Fig:tomo_0.5}
\end{figure*}

\begin{figure*}[htbp]
	\centering
	\includegraphics[width=1 \linewidth]{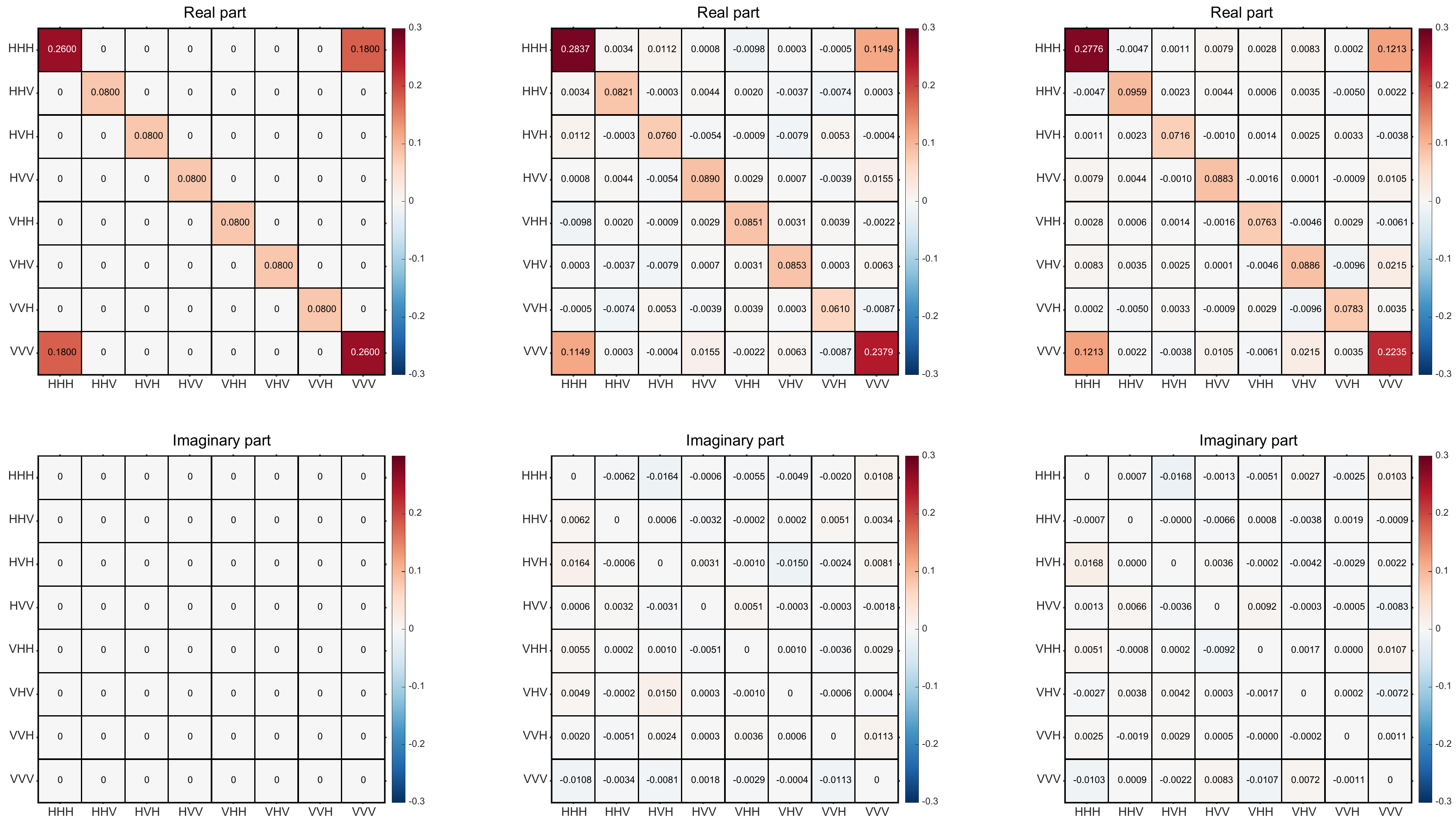}
	\caption{
\textbf{Left:} The density matrix of the noisy GHZ state with $p=0.36$. 
\textbf{Middle:} Experimental reconstructed density matrix of $\rho_{1}$ with $p=0.36$. 
\textbf{Right:} Experimental reconstructed density matrix of $\rho_{2}$ with $p=0.36$.
}
\label{Fig:tomo_0.36}
\end{figure*}

\newpage

%